%% file: main.tex
\documentclass{article}

\input{preamble2}

\input{macros}

\usepackage{xspace}
\newcommand{\FAST}{STEER}

\usepackage{comment}
\usepackage{mathtools}

\usepackage{authblk}

\title{Improving time dynamics simulation by sampling the error unitary}
\date{\today}

\author[1]{Lana Mineh}
\author[1]{Adrian Chapman}
\author[1]{Raul A. Santos}

\affil[1]{Phasecraft Ltd.}

\begin{document}

\maketitle
\thispagestyle{empty}

\begin{abstract}
We introduce an algorithm to improve the error scaling of product formulas by randomly sampling the generator of their exact error unitary.
Our approach takes an arbitrary product formula of time $t$, $S_k(t)$ with error $O(t^{k+1})$ and produces a stochastic formula with expected error scaling as $O(t^{2k+2})$ with respect to the exact dynamics. For a given fixed error $\epsilon$ and total evolution time $T$ this leads to an improved gate complexity of $N=O(T(T/\epsilon)^{\frac{1}{2k+1}})$ compared to the $O(T(T/\epsilon)^{\frac{1}{k}})$ gate complexity of a $k$-th order product formula.
This is achieved by appending an additional circuit with depth at-most logarithmic in the number of qubits, and without needing extra ancillas.
We prove that each instance of these stochastic formulas quickly concentrates to the expected value.
These results are based on an exact characterization of the error unitaries for product formulas. 
Through extensive numerical simulations we assess the performance of this approach in several spin and fermionic systems.
We show that the expansion of these error unitaries naturally leads to generalized Zassenhaus formulas, a result which could be of independent mathematical interest.
\end{abstract}

\section{Introduction}
\label{sec:intro}

Simulating dynamical properties of quantum systems is considered one of the most natural applications of quantum computers \cite{Feynman_1982}. 
Moreover, long time simulation of large quantum systems is expected to demonstrate quantum computational advantage over classical approaches, highlighting the sought-after separation between these computational paradigms \cite{Lloyd_2996}.

On a quantum computer, time-dynamics simulation (TDS) is achieved by approximating the evolution operator $U(T)=e^{-iTH}$ by a series of computational primitives, which could correspond to elementary gates or easy to prepare circuits. 
The gate complexity of implementing TDS scales linearly with the total evolution time $T$ \cite{Berry_2006,Childs_2009}, and several algorithms have been proposed to achieve that scaling \cite {Childs2010,Berry_2014,Low2017}. 
Among those approaches, product formulas \cite{Susuki1991,Ostmeyer_2023_trotter} are highly competitive in practice, despite their asymptotically worse gate complexity $O(T^{1+\frac{1}{2k}})$, where $2k$ is the order of the approximation.
This observation has motivated the search for improved algorithms for TDS based on product formulas by exploiting the structure of the Hamiltonian \cite{Bosse_2025, sharma2024}, numerically optimizing the constants \cite{yoshida1990construction,morales2022greatly} or numerically optimizing the full circuits \cite{Nam_2018,McKeever_2023_optimized}.
As product formulas give families of approximating unitaries, classical extrapolation methods can be used to reduce their circuit depth for a given target accuracy \cite{watson2024}. 
The circuit simplicity of product formulas has also made them the de facto choice for several cutting-edge quantum experiments on TDS where the limited coherence time of the device is a major limiting factor \cite{Kim_2023_evidence, Kivlichan2020,Karacan2025, vilchez2025}.

In this context, randomized algorithms offer a particularly appealing approach, as they provide a natural way to reduce circuit complexity by trading it (naively) for an increased sample complexity.
Randomized protocols to boost the error of TDS algorithms have been widely studied in the literature under different guises.
In \cite{Childs_2019_faster} the authors proved that a randomized application of different orderings in the product formulas improves the gate complexity from $O(TL^2(TL/\epsilon)^{\frac{1}{k}})$ to $\max\left\{O(TL^2(TL/\epsilon)^{\frac{1}{2k+1}}),O(TL^2(T/\epsilon)^\frac{1}{k})\right\}$ for formulas composed of $L$ summands, with target error $\epsilon$ and total evolution time $T$. 
Making use of an improved mixing lemma \cite{Chen2021}, it is possible to further improve this result to $O(TL^2(T/\epsilon)^{\frac{1}{k}})$ since, in fact, fewer gates are needed to achieve the same error scaling.
In \cite{Campbell_2019} the author introduced the qDRIFT method which approximates the evolution generated by $U(T) = e^{-iTH}$ using a stochastic quantum channel given by randomly sampling unitaries generated by the individual terms $H_j$ appearing in $H \coloneqq \sum_j\alpha_j H_j$. 
This approach, although scaling as a first-order formula with a gate complexity $O(\Lambda^2T^2/\epsilon)$ for a target error $\epsilon$, has a dependence on the $\ell_1$ norm of the Hamiltonian $\Lambda = \sum_i|\alpha_i|$, instead of its maximum. 
This is especially useful in settings where the Hamiltonian has many terms of varying size \cite{WanPRL2022}. 
The concentration properties of this method have been discussed in \cite{Chen2021}.
Extensions of this stochastic method that achieve better error scaling with time have been discussed in \cite{Nakaji_2024}, but they usually rely on controlled evolutions to implement the cross terms in the Taylor expansion.

In this work, we introduce a novel algorithm that quadratically improves the error scaling of a given product formula $S_k(t)$ from $O(t^{k+1})$ to $O(t^{2k+2})$ with minimal added circuit complexity and without needing ancillas.
Dividing a total evolution time $T$ in $N$ segments and for a given target accuracy $\epsilon$, this leads to an improved gate complexity of $N=O(T(T/\epsilon)^{\frac{1}{2k+1}})$ compared to the $O(T(T/\epsilon)^{\frac{1}{k}})$ gate complexity of a $k$-th order product formula.
Crucially, this approach requires a circuit depth that is lower than the corresponding product formula of the same accuracy.
This improvement is possible via an exact characterization of the error $F(t)=S_k^\dagger(t)U(t)$ -- relating the product formula to the exact evolution $U(t)$ -- given as the exponential of a time-dependent Hamiltonian. 
This result is shown in \cref{thm:error_rep}. 
The lowest-order contribution of this effective Hamiltonian scales as $t^{k+1}$, and this means that a usual first-order approximation of the unitary it generates will incur an error $O(t^{2k+2})$. 
Although any known approximation for time-dependent dynamics could be used to further approximate $F$, the structure of the time dependent Hamiltonian could be very involved due to the multiply nested commutators between the generators appearing in $S_k(t)$. 
This motivates the introduction of a stochastic approach, where randomly selecting Paulis from the expansion of the effective time-dependent Hamiltonian approximates the error $F(t)$ in expectation. 
We prove that the error of this approximation scales as $O(t^{2k + 2})$ in \cref{thm:approx_rand}. 

In \cite{Zeng_2025} the authors propose an algorithm that improves the error of a $k$-th order product formula to order $2k$ by implementing controlled operations on an ancilla qubit to reduce the error of product formulas via linear combination of unitaries (LCU). 
While this approach achieves the same scaling as our algorithm, it uses ancillas to implement the cross terms in the LCU expansion.
In \cite{peetz_2025}, a novel ancilla-free method based on the Zassenhaus expansion is shown to allow for an arbitrarily small error. This algorithm makes use of the Zassenhaus formula \cite{Magnus_1954}
\begin{align}\label{eq:zassenhaus}
 e^{it(A+B)}=e^{itA}~e^{itB}~e^{{\frac {t^{2}}{2}}[A,B]}~e^{-i{\frac {t^{3}}{6}}(2[B,[A,B]]+[A,[A,B]])}~e^{{\frac {-t^{4}}{24}}([[[A,B],A],A]+3[[[A,B],A],B]+3[[[A,B],B],B])}\cdots.
\end{align}
and its generalization to include many summands.
Truncating the product on the right at the exponential parameterized by $t^k$ generates an approximation of $U$ with error $O(t^{k + 1})$. 
The nested commutators generating each exponential make their implementation challenging, as we require further product formulas to approximate them in principle. 
The number of exponentials required to approximate the commutators does not scale very favorably in general \cite{Chen_2022}. 
Assuming that they can implement the exponentials up to the $t^{k-1}$ term, the authors rely on a first order qDRIFT approximation of the exponential parameterized by $t^k$, which scales as $O(t^{2k})$, thus allowing to extend their method to achieve an error $O(t^{2k})$.

A similar approach to the one we develop here has been presented in ref.~\cite{Cho_2024}, though we offer several technical improvements over that method. 
The main difference is that we propose a {\it multiplicative} unitary model of the error $U=S_kF$ where $U$ is the exact unitary, $S_k$ is the approximant and $F$ is the error, instead of the additive model $U =S_k + F$ discussed in \cite{Cho_2024}. 
This allows us to deduce an explicit construction based on a given approximant, something not explicitly constructed in \cite{Cho_2024}. 
Moreover, our method generalizes the proposal of ref.~\cite{Cho_2024} to non-symmetric product formulas, but includes symmetric product formulas as a special case. 
We also discuss the error incurred by our method in the finite-sampling case using the concentration results of ref.~\cite{Chen2021}, and we additionally use the improved mixing lemma from that reference to make improvements throughout our analysis.
To assess the performance of our approach we perform extensive numerical simulations in different models.
Finally, our construction generalizes Zassenhaus formulas, something that could be of independent mathematical interest.

The additional accuracy of our approach comes at the expense of the additional circuit depth of the stochastic part, which is upper bounded by the logarithm of the number of qubits as we discuss. This also explains why it is not possible (with this approach) to achieve an improvement beyond $O(t^{2k+2})$ for a base product formula $S_k(t)$.

We discuss the application of our results in two complementary computational models. The first corresponds to the case when the answer is obtained from a single application of the algorithm, where phase estimation is a classic example. 
The second corresponds to the case when repeated runs of the algorithm are needed, where the estimation of expectation values is a natural example. Bounding the diamond norm of the difference between our approach and the exact evolution, we prove in 
\cref{thm:approx_rand} that in the first computational model we indeed achieve an error scaling as $O(t^{2k+2})$, thus matching the performance of a $S_{2k+2}(t)$ product formula, just with a $S_{2k}(t)$ approximant decorated with a circuit of depth at most logarithmic in the number of qubits. 

Using the theory of matrix martingales \cite{Tropp2011}, we are able to extend the results of \cite{Chen2021} and \cite{Kiss2023} in our context, to prove the concentration properties of our approach in the second computational model. In particular we show that each instance of the sampled approximation channel becomes closer to the expected value as the number of Trotter steps increase. Specifically we show that the error of our approach is bounded by $\max\left(O(N(T/N)^{2k+2}),O(\sqrt{N}(T/N)^{k+1}/M)\right)$, where $N$ is the number of Trotter steps, $T$ is the total evolution time and $M$ is the number of measurements. 
This implies that using a $k-$th order approximant and sampling from the error unitary we achieve a gate complexity for a fixed error $\epsilon$ of
\begin{align}
    N = O\left(T\left(\frac{T}{\epsilon^{2}}\right)^{\frac{1}{2k+1}}\right)\max\left(O\left(\frac{1}{M^{\frac{1}{2k+1}}}\right),O\left(\epsilon^{\frac{1}{4k+2}}\right)\right)
\end{align}  
a result that appears in \cref{coro:fluct}. This represents an exponential improvement in circuit depth as a function of $k$ compared with the equivalent product formula, as we discuss in \cref{sec:depth_diss}. The practical effectiveness of our approach (which we dub \FAST{}), is shown by performing extensive numerical simulations for a variety of spin, electronic, many-body and molecular systems. In \cref{fig:heisenberg_target_error_scaling} we show the scaling of gate complexity for simulating the one dimensional Heisenberg model with random fields, for an evolution time that scales with the system size, where we fix the target error to $\epsilon = 10^{-3}$. We compare two different sampling protocols for \FAST{} that are naturally suited for this, which are further described in section \cref{sec:numerics}. Note the 10x reduction in circuit resources compared to the standard second order formula.

Our work is structured as follows: In \cref{sec:motivation} we present the general intuition and motivation of our approach. 
The main analysis of our algorithm, technical results and the proofs of \cref{thm:error_rep} to \cref{coro:fluct} are presented  in \cref{sec:analysis}. 
Generalizations of our approach, in terms of different sampling strategies, general error formulas and a generalized Zassenhaus formula is  discussed in \cref{sec:generalizations}. 
The performance of our approach is assessed in practice through numerical results in \cref{sec:numerics}. Finally, we discuss the implications some implications of this method in \cref{sec:discussion}.

\section{Summary of Results}

\begin{figure}
    \centering
    \includegraphics[width=0.5\textwidth]{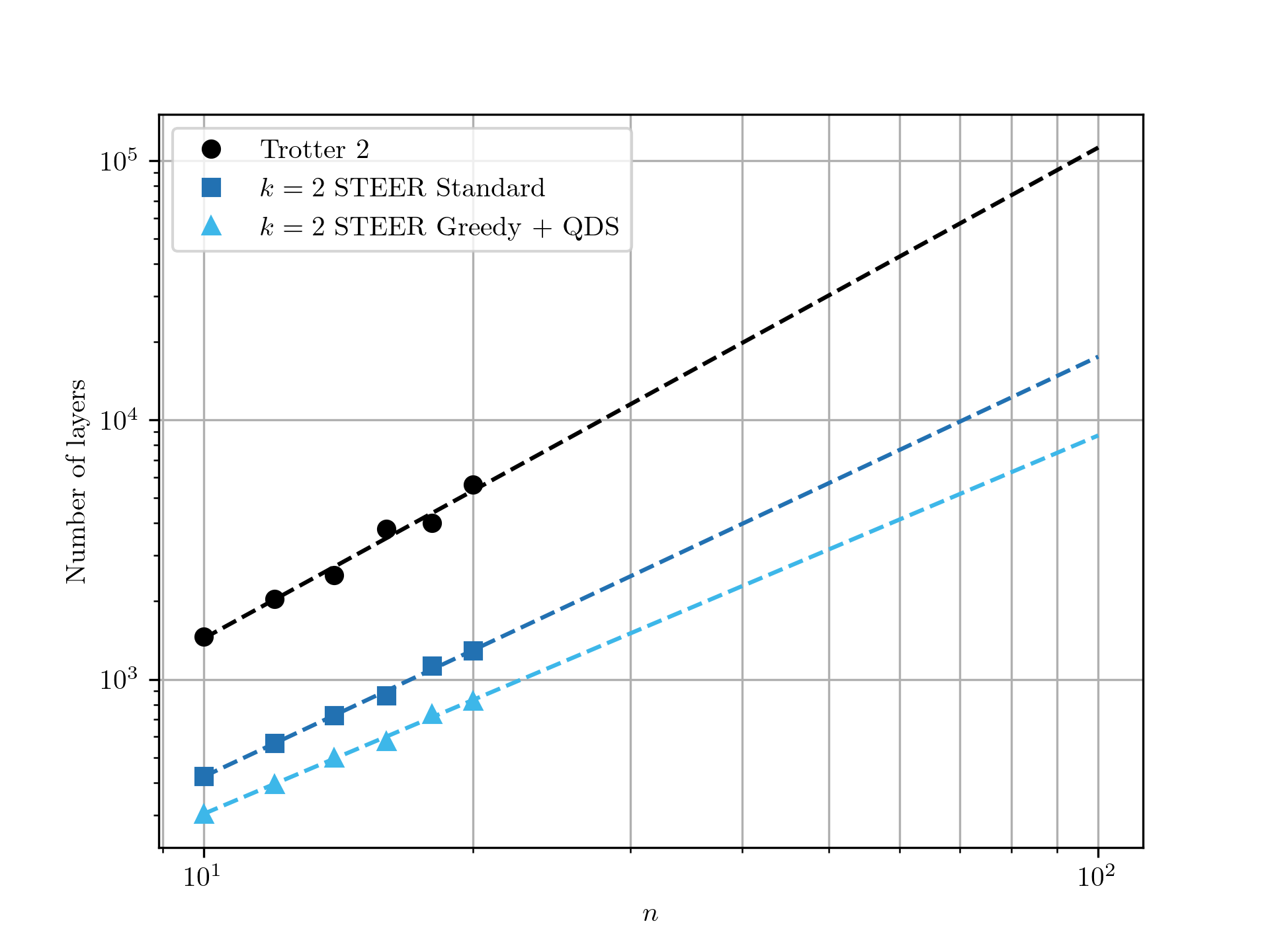} 
    \caption{Number of layers to reach a target error of $10^{-3}$ for $1 \times n$ Heisenberg model evolved up to a time scaling with the system's size $t = n$. The number of STEER samples taken was 10,000. The lines of best fit are for Trotter 2: $18.33n^{1.89}$, \FAST{} standard: $10.08n^{1.62}$ and \FAST{} greedy + qubit disjoint sampling (QDS): $10.47n^{1.46}$.
    }
    \label{fig:heisenberg_target_error_scaling}
\end{figure}

Our theoretical results are concisely stated in terms of three theorems
\begin{innercustomresult}[Informal statement of Theorems \ref{thm:error_rep} and \ref{thm:approx_rand}]\label{result:thms1and2} Given a $k$th-order product formula $S_{k}(t)$ approximating the unitary $U=e^{itH}$ with error $O(t^k)$, there exists a quantum channel $\mathcal{Z}_{2k+1}$ acting as
\begin{align} \label{eq:zchanneldef}
    \mathcal{Z}_{2k+1}(\rho) :=\sum_{\ell}p_{\ell}(t){S}_{k}(t)V_{\ell}\rho V_{\ell}^\dagger{S}^\dagger_{k}(t)
\end{align}
approximating the exact evolution channel $\mathcal{U}(\rho) = U \rho U^{\dagger}$ with error
\begin{align} \label{eq:1steperror}
    \frac{1}{2}\|\mathcal {U} -\mathcal{Z}_{2k+1} \|_\diamond \leq C t^{2k+2}
\end{align}
for $V_{\ell} = e^{i \theta_{\ell}(t) P_{\ell}}$, a random Pauli gate sampled with probability $p_{\ell}(t)$, and $C$ a constant independent of $t$. 
The ensemble $\mathcal{E} \coloneqq \{(p_{\ell}, V_{\ell})\}_{\ell}$ is determined from terms and coefficients in the error Hamiltonian, $A_{k}(t)$, which generates $F(t) = S_k^{\dagger}(t)U(t)$ and scales as $O(t^{k + 1})$.
\end{innercustomresult}
We extend this result to the setting of $N>1$ Trotter steps and include fluctuations in the sample mean approximating $\mathcal{Z}^N_{2k + 1}$ in the following concentration bound
\begin{innercustomresult}[Informal statement of Theorem \ref{theo:concentration}, Corollary \ref{coro:fluct}]
Consider a collection of $M$ ordered configurations, each of $N$ independent samples from the ensemble $\mathcal{E}$ described in the previous result, and let 
\begin{align}
    \mathcal{Z}^N_{2k + 1, m}(\rho) \coloneqq \prod_{j = 1}^N S_k(t) V_{m, j} \rho V_{m, j}^{\dagger} S_{k}^{\dagger}(t)
\end{align} denote the random unitary channel corresponding to the $m$th configuration for $1 \leq m \leq M$, where $V_{m, j}$ is the $j$th unitary sampled in this configuration.
We have
\begin{align}\label{eq:scalingtradoff}
\mathbb{E}\left[\frac{1}{2}\left\|\mathcal{U}^N - \frac{1}{M}\textstyle\sum_{m = 1}^M \mathcal{Z}_{2k + 1, m}^{N} \right\|_\diamond\right]&\leq C N t^{2k+2}+ C' (k + 1) \tilde{\lambda} t^{k+1} \sqrt{\frac{N \mathrm{log}(2^{n + 1})}{M}}
\end{align}
where $C$ is the constant from \cref{eq:1steperror}, and $C'$ is an independent numerical constant.
Here, $\tilde{\lambda}$ is the largest-magnitude Pauli coefficient in the error Hamiltonian described in \cref{result:thms1and2}, and $n$ is the number of qubits. 
\end{innercustomresult}
To give some intuition for this scaling, we provide a graphical depiction of \cref{eq:scalingtradoff} in \cref{fig:scaling_sketch}.
Fixing $N$, $k$, and the Hamiltonian-dependent constants, $n$ and $\tilde{\lambda}$, we sketch the logarithm of the upper bound as a function of $\mathrm{log}(t)$ and the number of configuration samples $M$.
As the the logarithm of this bound is very close to the logarithm of the maximum of the two terms, we expect that, for small timestep $t$, the error scales as $O(M^{-1/2} t^{k + 1})$, and this scaling transitions to $O(t^{2k + 2})$ at larger times. 
Additionally, increasing the number of configuration samples $M$ shifts the transition point to a smaller value of $t$.

\begin{figure}
    \centering
    \includegraphics[width=0.6\linewidth]{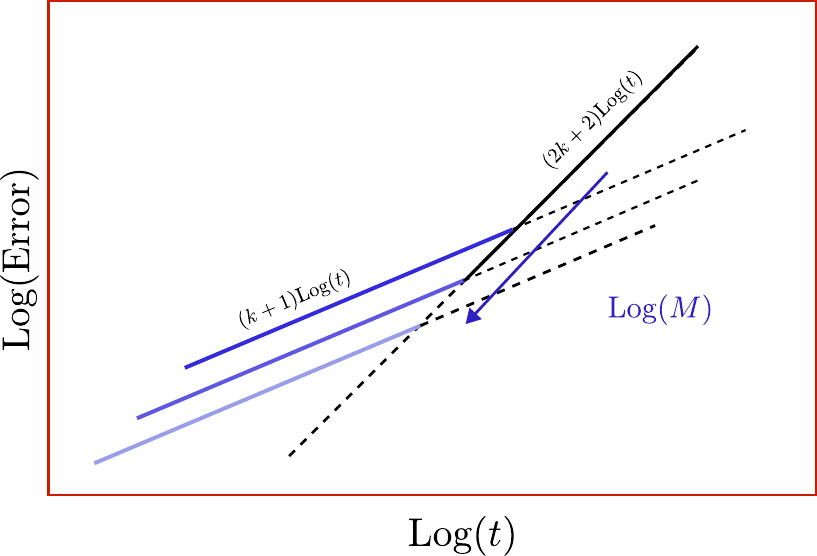}
    \caption{Sketch of algorithm's scaling with timestep $t$ and the number of configuration samples $M$, according to \cref{eq:scalingtradoff}. 
    As discussed above, for sufficiently small timestep $t =T/N$, the error is controlled by the second term in \cref{eq:scalingtradoff}, this produces a scaling with a slope $(k+1)$ in a log-log plot. 
    The transition point between the $k+1$ and the $2k+2$ scalings is controlled logarithmically by number of measurements. 
    We produce this type of scaling plot for different specific models in \cref{sec:numerics}.}
    \label{fig:scaling_sketch}
\end{figure}

\section{Motivation}
\label{sec:motivation}

As a motivation for the general approach we start from the simple decomposition of a Hamiltonian into two terms $H=A+B$ and define $e^{it(A+B)}:=e^{itA}V$, where we are after a description of the error unitary $V$. We will discuss the generalization to an arbitrary number of terms in \cref{sec:analysis}. Taking this equation as a definition of $V$ we find 
\begin{align}
\frac{dV}{dt}	=i\left(e^{-itA}(A+B)e^{itA}-iA\right)V
	=i\left(e^{-itA}Be^{itA}\right)V.
\end{align}
The formal solution for $V$ is $V=\mathcal{T}e^{i\int_{0}^{t}B(s)ds}$ where $\mathcal{T}$ is the time ordered operator and $B(t):=e^{-itA}Be^{itA}$. An expansion in powers of $t$ of $B(t)$ is given by $B(t)=B-it[A,B]-\frac{t^2}{2}[A,[A,B]]+O(t^3)$, from where we find
\begin{align}\nonumber
e^{it(A+B)}=e^{itA}\mathcal{T}e^{i\int_0^t B(s)ds}&=e^{itA}e^{i\int_0^t B(s)ds}+O(t^3)\\
&=e^{itA}e^{itB}e^{-i\frac{t^2}{2}[A,B]}+O(t^3).
\end{align}
This particular procedure generates the known Zassenhaus formula \cite{Magnus_1954} up to $O(t^3)$.

This construction is very general, as we can generate novel expansions just by choosing different starting formulas from where the error unitary is characterized. Starting from the second order Trotter formula, we have for example for an expansion of the Hamiltonian in two terms (arbitrary number of terms are discussed in \cref{sec:general_partition})
\begin{align}\label{eq:Second_order}
    e^{it(A+B)}	&=e^{i\frac{t}{2}A}e^{itB}e^{i\frac{t}{2}A}\mathcal{T}e^{i\int_{0}^{t}A_{2}(s)ds}\\
    &=e^{i\frac{t}{2}A}e^{itB}e^{i\frac{t}{2}A}e^{i\int_{0}^{t}A_{2}(s)ds}+O(t^6)\\
    &=e^{i\frac{t}{2}A}e^{itB}e^{i\frac{t}{2}A}e^{-i\frac{t^{3}}{24}\left([A+2B,[A,B]]\right)}+O(t^{4})
\end{align}
where
\begin{align}
    A_{2}(t)&=e^{-i\frac{t}{2}A}\left(\frac{1}{2}\left(e^{-itB}Ae^{itB}-A\right)+e^{-itB}\left(e^{-i\frac{t}{2}A}Be^{i\frac{t}{2}A}-B\right)e^{itB}\right)e^{i\frac{t}{2}A}\\
    &=-\frac{t^{2}}{8}\left([A,[A,B]]+2[B,[A,B]]\right)+O(t^{3}).
\end{align}
For the general expression and the explicit expansion up to $O(t^4)$ of $A_2(t)$ see \cref{app:genA2}.
The Zassenhaus-like formula produced by this approach is arguably simpler than the same order approximation of the original formula \cref{eq:zassenhaus}. More importantly, \cref{eq:Second_order} generates an improved approximation of $U$ with error controlled by the implementation of  $\mathcal{T}e^{i\int_{0}^{t}A_{2}(s)ds}$. Motivated by the constraint of reducing circuit depth in actual implementations we approximate the error unitary $\mathcal{T}e^{i\int_{0}^{t}A_{2}(s)ds}$ stochastically to first order by sampling the Paulis appearing in the expansion of the effective Hamiltonian $A_2(t)$ that generates it.

We can repeat this procedure starting with any product formula, generating improved formulas with better error scaling. To do this, the idea is simple:
\begin{enumerate}
    \item Characterize the error unitary $F(t)$ with respect to $U=e^{iHt}$ in a $k$-th order product formula $S_k$ as
    \begin{align}
        U(t) = S_k(t) F(t)
    \end{align}
    where $F(t)$ is a exact description of the error. This characterization is the main result of \cref{thm:error_rep}.
    
    \item Generate an approximation of $F(t)$ that is efficiently implementable with error $O(t^{k+2})$ or better. This approach is captured in \cref{thm:approx_rand}, where we approximate the exponential $F(t)$ with a stochastic formula, leading to an error scaling as $O(t^{2k+2})$.
\end{enumerate}

This insight naturally generates an algorithm for approximating time evolution with improved error and circuit depth compared to the corresponding product formulas or qDRIFT approaches, without requiring ancillas. 
We call this approach Sampling the Trotter Error for Enhanced Resolution (\FAST{}). The pseudocode for one variant of the algorithm is shown in \cref{algo:var_general}.

In the rest of this work we develop this idea further, proving the performance of the \FAST{} algorithm in some of its variants, its concentration properties and analyzing its practical potential in some test examples.

\section{Analysis of the \FAST{} algorithm}
\label{sec:analysis}

\subsection{\FAST{}, arbitrary number of summands in $H$}\label{algo:var_general}

The simplest form of \FAST{} is shown in \cref{algo:gen}, applied to a Hamiltonian with an arbitrary number of summands. There are different ways of define the sampling of unitaries in practice from the error representation, where considerations about the resources available for implementations can be included. We further discuss this point in \cref{sec:generalizations}. In \cref{algo:gen}, we consider the sampling protocol where any term in the effective Hamiltonian can be sampled. This approach has the advantage that it requires the shallowest circuit, as we discuss in \cref{sec:depth_diss}. Other approaches are discussed in \cref{app:other_vars}.

In the rest of this section we analyze the performance of this algorithm in terms of the error that it achieves and the fluctuations of each random realization of \FAST{} away from the expected channel.  

\begin{algorithm}[H]
\SetAlgoLined
 \textbf{Input:} Given a Hamiltonian $H=\sum_{j=1}^l h_j$, where $h_j$ are operators such that $e^{ith_j}$ is directly implementable, a $k$-th order product formula $S_k^{(l)}(t)$ and evolution time $t$ \\
\KwResult{Stochastic approximant of $e^{itH}$ with error of order $O(t^{2k+2})$.}

\medskip

 Compute the series expansion in time of 
 \begin{align}
     A_{k}(t):=S_k^{\dagger}(t) H S_k(t)-i\frac{d}{dt}(S_k^{\dagger}(t))S_k(t),
 \end{align} up to order $t^{2k+1}$, i.e $A_k^{(l)}(t)=\sum_{m=k}^{2k}t^m\Omega_m +O(t^{2k+1})$

\For{$j=k\dots2k$}{
Compute $\Omega_j$ as a sum of Pauli terms $\Omega_j=\sum_{m=1}^{w}\alpha_{jm}P_m$\\
Output the probability distribution $p_{m|j}:=\frac{|\alpha_{jm}|}{\sum_{p=1}^{r_{\rm max}}|\alpha_{jp}|}$\\
Output the probability distribution $p_j(t):=\frac{t^{j-k}}{j+1}\left(\sum_{s=0}^k\frac{t^s}{k+1+s}\right)^{-1}$ 
}

 Sample the random variable $M\in\{k,\dots,2k\}$ from the probability distribution $p_m(t)$\\
 Sample the random variable $R$ from the probability distribution $p_{m|r}$, i.e. ${\rm Pr}(R=r|M = m)=p_{m|r}$.\\
 Construct the estimator
 $X(t)=S_k^{(l)}(t)e^{i\theta_{mr}(t)P_r}$
 where $\theta_{mr}(t)={\rm sign}(\alpha_{mr})\left(\sum_{s=k}^{2k}\frac{t^{s+1}}{s+1}\right)\sum_{p=1}^{r_{\rm max}}|\alpha_{mp}|$
 
 Output $X(t)$ 
 \caption{Approximating time evolution with stochastic formulas}\label{algo:gen}
\end{algorithm}

\subsection{Error bounds}

The analysis of the error of \FAST{} is simplified using \cref{thm:error_rep}, which establishes a useful representation of the error incurred by a product formula.

\begin{theorem}\label{thm:error_rep}
Given some $k$th order product formula $S_{k}(t)$ that approximates the unitary $U=e^{itH}$ with error $O(t^{k+1})$, the time dependent Hamiltonian
\begin{align}\label{eq:A_k}
    A_{k}(t):=-i\frac{d}{dt}(S_k^{\dagger}(t)U)U^\dagger S_k(t)=S_k^{\dagger}(t) H S_k(t)-i\frac{d}{dt}(S_k^{\dagger}(t))S_k(t),
\end{align}
 generates the unitary
\begin{align}\label{eq:definition_err}
    U=S_{k}(t)\mathcal{T}e^{i\int_0^tA_{k}(s)ds},
\end{align}
moreover, $\int_0^t\|A_{k}(s)\|ds=O(t^{k+1})$.
\end{theorem}
\begin{proof}
The proof is straightforward. Define $F(t):=S_k^{\dagger}(t)U$. $F$ satisfies the first order ordinary differential equation \begin{align}\nonumber\label{eq:solution_TO}
    \frac{dF}{dt}&=\frac{d}{dt}(S_k^{\dagger}(t)U)=\left(\frac{d}{dt}(S_k^{\dagger}(t)U)\right)U^\dagger S_k(t)F(t),\\
    &=i\left[-i\frac{d}{dt}(S_k^{\dagger}(t)U)U^\dagger S_k(t)\right]F=iA_k(t)F(t).
\end{align}
The solution of this equation is given by $F(t)=\mathcal {T}e^{i\int_0^tA_k(s)ds}$, from where it follows that $U=S_k(t)F(t)=S_k(t)\mathcal {T}e^{i\int_0^tA_k(s)ds}$. 
The operator $A_k(t)$ is Hermitian, and as such can be considered a (time dependent) Hamiltonian. This follows from 
\begin{align}
A_k^\dagger(t)&=iS_k^{\dagger}(t)U\frac{d}{dt}(U^\dagger S_k(t))=i\frac{d}{dt}\left(S_k^{\dagger}(t)UU^\dagger S_k(t)\right)-i\frac{d}{dt}(S_k^{\dagger}(t)U)U^\dagger S_k(t)\\\label{eq:hermitian}
&=-i\frac{d}{dt}(S_k^{\dagger}(t)U)U^\dagger S_k(t)=A_k(t).
\end{align}
Finally, the scaling of $A_k(t)$ with time follows from the integral representation of the error
\begin{align}
E(t):=F(t)-1&=S_k^{\dagger}(t)U-1 = \int_0^tdt' \,\left(\frac{dF}{dt'}\right)\quad \mbox{by the definition of $F(t)$}\label{eq:rep_E_T}
\end{align}
and the initial condition $F(0)=1$.
This leads to
\begin{align}
    E(t)=\int_0^tiA_k(t')\mathcal{T}e^{i\int_0^{t'} A_k(s)ds}dt'.
\end{align}
Taking the operator norm and using the triangle inequality leads finally to $\|E(t)\|\leq  \int_0^t \|A_k(s)\|ds $. As $S_k(t)$ is a product formula of order $k$, $\|E(t)\|=O(t^{k+1})$.
\end{proof}

The exact implementation of the unitary operator $\mathcal{T} e^{i\int_0^t A_k(s)ds}$ could be difficult in practice. It is then convenient to develop a strategy to approximate this unitary in terms of simple computational primitives. We achieve this by first computing $A_k(t)$ as a power series in the time variable, adn then randomly sampling terms from that expansion.

Concretely, assuming that $\int_0^tA_k(s)ds$ has an expansion in terms of computational primitives $P_k$, which wlog we can assume to be Pauli operators, i.e.
\begin{align}\label{eq:expansion}
    \int_0^tA_k(s)ds=t^{k+1}\sum_{m=0}^\infty \frac{t^{m}}{m + 1 + k}\sum_{r=1}^{r_{\rm max}}\alpha_{mr} P_r,
\end{align}
we propose approximating $e^{i\int_0^t A_k(s)ds}$ by the expectation of a random unitary $V$ defined as:
\begin{definition}\label{def_U}
     Given a parameter $t>0$ and a rectangular real matrix with entries $\alpha_{mr}$, the discrete random variables $M\in\{0,\dots,k\}$ and $R\in\{1,\dots, r_{\max}\}$ are defined by their probability densities as 
    \begin{align}
        {\rm Pr}(M=m)&:=p_m(t):=\frac{t^m}{(m + 1 + k)\Lambda(t)},\\
        {\rm Pr}(R=r|M=m)&:=p_{r|m}:=\frac{|\alpha_{mr}|}{\lambda_m},
    \end{align} 
    where $\Lambda(t):=\sum_{l=0}^k \frac{t^l}{k + 1 + l}$ and $\lambda_m:=\sum_{r=1}^{r_{\rm max}}|\alpha_{mr}|$. 
    Similarly, the random unitary $V\in\{\{e^{i\theta_{mr}(t)P_r}\}_{r=1}^{r_{\rm max}}\}_{m=0}^k$ is defined by the probability density
    \begin{align}
        P(V=e^{i\theta_{mr}(t)P_r})=p_m(t)p_{r|m}.
    \end{align}
    Here $\theta_{mr}(t) =t^{k+1}{\rm sign}(\alpha_{mr})\Lambda(t)\lambda_m $ and
    $P_r$ are the operators that appear in the expansion \cref{eq:expansion}.
\end{definition}  
  Expanding the expectation of the random unitary $V$ to lowest order in $t$ we have
\begin{align}\nonumber
  \mathbb{E}[V]=\sum_{m=0}^kp_m(t)\sum_r p_{r|m}e^{i\theta_{mr}(t)P_r}&=1+it^{k+1}\sum_{m=0}^\infty \frac{t^{m}}{m + 1 +   k}\sum_{r=1}^{r_{\rm max}}\alpha_{mr} P_r +O(t^{2k+2})\\
  &=\mathcal{T}e^{i\int^t_0 A_k(s)ds}+O(t^{2k+2})
\end{align}
A precise bound on the error of this approach is captured in \cref{thm:approx_rand} below, where we bound the diamond norm between the exact evolution and \FAST{}.
\begin{theorem}\label{thm:approx_rand}
    Given $0\leq t\leq  1$, $\int_0^tA_k(s)ds=t^{k+1}\sum_{m=0}^\infty \frac{t^{m}}{m + 1 + k}\sum_r\alpha_{mr} P_r$, the stochastic channel $\mathcal{Z}_{2k+1}(\rho):=\sum_{m,r}p_m(t)p_{r|m}{S}_{k}(t)V_{mr}\rho V_{mr}^\dagger{S}^\dagger_{k}(t)$ where $V_{mr}=e^{i\theta_{mr}(t)P_r}$ is a random unitary as described in \cref{def_U} approximates the exact evolution channel $\mathcal{U}(\rho)=U\rho U^\dagger$ with an error
    \begin{align}
    \frac{1}{2}\|\mathcal {U} -\mathcal{Z}_{2k+1} \|_\diamond \leq C t^{2k+2}
    \end{align} 
    where $C$ is a constant independent of $t$.
\end{theorem}

\begin{proof}
    Using the relation between the operator and diamond norm \cite{Chen2021} (Lemma 3.4), unitary invariance of the operator norm, and \cref{thm:error_rep} above, we have
    \begin{align}
        \frac{1}{2}\|\mathcal {U} -\mathcal{Z}_{2k+1} \|_\diamond &\leq \|U -{S}_k\mathbb{E}[V]\|= \|\mathcal{T}e^{i\int_0^tA_k(s)ds} - \mathbb{E}[V] \|=\|\mathbb{E}[ V]\mathcal{T}e^{-i\int_0^tA_k(s)ds}-1 \|\\
        &=\left\Vert \sum_{m=0}^{k}p_m(t)\sum_{r}p_{r|m}e^{i\theta_{mr}(t)P_{r}}\mathcal{T}e^{-i\int_{0}^{t}A_{k}(s)}-1\right\Vert, \quad\mbox{using the definition of $\mathbb{E}[V]$} 
    \end{align}
    This motivates the introduction of the error function $E:=\sum_{m=0}^{k}p_m(t)\sum_{r}p_{r|m}e^{i\theta_{mr}(t)P_{r}}\mathcal{T}e^{-i\int_{0}^{t}A_{k}(s)ds}-1$, which has an integral representation $E=\int_0^t\mathfrak{S}(\tau)\mathcal{T}e^{-i\int_{0}^{\tau}A_{k}(s)ds}d\tau$ where
    \begin{align}\label{eq:sum_inside}
        \mathfrak{S}(\tau)
        &:=\sum_{m=0}^{k}p_m(\tau)\sum_{r}p_{r|m}e^{i\theta_{mr}(\tau)P_{r}}\left(\frac{\dot{p}_m(\tau)}{p_m(\tau)}+i\left(\dot{\theta}_{mr}(\tau)P_{r}-A_{k}(\tau)\right)\right)
    \end{align}
           and $\dot{f}:=\frac{df}{d\tau}$. We can simplify this expression making use of the expansion
           \begin{align}
               \int_{0}^{t}A_{k}(s)ds&=t^{k+1}\sum_{m=0}^{\infty}\frac{t^{m}}{m + 1 + k}\sum_{r=1}^{r_{{\rm max}}}\alpha_{mr}P_{r} \\
               &=\sum_{m=0}^{k}p_m(t)\sum_{r}p_{r|m}\theta_{mr}(t)P_{r}+t^{k+1}\sum_{m=k+1}^{\infty}\frac{t^{m}}{m + 1 + k}\sum_{r=1}^{r_{{\rm max}}}\alpha_{mr}P_{r}
           \end{align}
           from where it follows that
           \begin{align}
               A_k(\tau)=\sum_{m=0}^{k}p_m(\tau)\sum_{r}p_{r|m}\left[\frac{\dot{p}_m(\tau)}{p_m(\tau)}\theta_{mr}(\tau)+\dot{\theta}_{mr}(\tau)\right]P_{r}+\frac{d}{d\tau}\left(\tau^{k+1}\sum_{m=k+1}^{\infty}\frac{\tau^{m}}{m + 1 + k}\sum_{r=1}^{r_{{\rm max}}}\alpha_{mr}P_{r}\right).
           \end{align}
Inserting this back in \cref{eq:sum_inside} leads to
\begin{align}\label{eq:sum_improved}
\mathfrak{S}(\tau)=&-i\frac{d}{d\tau}\left(\tau^{k+1}\sum_{m=k+1}^{\infty}\frac{\tau^{m}}{m + 1 + k}\sum_{r=1}^{r_{{\rm max}}}\alpha_{mr}P_{r}\right)\\\nonumber
&+\sum_{m=0}^{k}p_m(\tau)\sum_{r}p_{r|m}\left(\frac{\dot{p}_m(\tau)}{p_m(\tau)}(e^{i\theta_{mr}(\tau)P_{r}}-i\theta_{mr}(\tau)P_{r}-1)+i\left(\dot{\theta}_{mr}(\tau)P_{r}-A_{k}(\tau)\right)(e^{i\theta_{mr}(\tau)P_{r}}-1)\right).
\end{align}
\cref{eq:sum_improved} allows us to bound the error directly as
\begin{align}
    \| E\|&\leq \int_0^t \| \mathfrak{S}(\tau)\|d\tau \leq \int_0^t \left\|\frac{d}{dt}\left(t^{k+1}\sum_{m=k+1}^{\infty}t^{m}\sum_{r=1}^{r_{{\rm max}}}\alpha_{mr}P_{r}\right)\right \|d\tau\\\nonumber
    &+\int_0^t\sum_{m=0}^{k}p_m(\tau)\sum_{r}p_{r|m}\left(\frac{|\dot{p}_m(\tau)|}{p_m(\tau)}\|e^{i\theta_{mr}(\tau)P_{r}}-i\theta_{mr}(\tau)P_{r}-1\|+\left\|\dot{\theta}_{mr}(\tau)P_{r}-A_{k}(\tau)\right\|\|e^{i\theta_{mr}(\tau)P_{r}}-1\|\right)d\tau
\end{align}
where we have used the triangle inequality. Using the bound $\| e^{ix}-ix-1\|\leq \frac{1}{2}\|x^2\|$ and $\| e^{ix}-1\|\leq\|x\|$ valid for $x$ an Hermitian matrix, the bound becomes
\begin{align}
     \| E\| \leq t^{k+1}\sum_{m=k+1}^{\infty}t^{m}\sum_{r=1}^{r_{{\rm max}}}|\alpha_{mr}|
    +\int_0^t\sum_{m=0}^{k}p_m(\tau)\sum_{r}p_{r|m}\left(\frac{1}{2}\frac{|\dot{p}_m(\tau)|}{p_m(\tau)}|\theta^2_{mr}(\tau)|+\left\|\dot{\theta}_{mr}(\tau)P_{r}-A_{k}(\tau)\right\||\theta_{mr}(\tau)|\right)d\tau
\end{align}
where we have used wlog that $\|P_r\|=1$. Here we see that the error is composed of three terms. The first one indicates that $\int_0^tA_k(s)ds$ is approximated up to order $O(t^{2k+1})$. In contrast, the first term inside the integral depends on the rate of variation of the probability distribution $p_m(t)$, while the last term depends on the way we sample the Paulis appearing in $A_k(t)$. Using the particular choices $p_m(t)=t^m/\Lambda(t)$ and $|\theta_{mr}|=\Lambda(t)\lambda_{m}t^{k+1}$, we can simplify this bound further to
\begin{align}
    \|E\|\leq t^{k+1}\sum_{m=k+1}^{\infty}t^{m}\lambda_{m}+\int_{0}^{t}\sum_{m=0}^{k}\tau^{m}\lambda_{m}\tau^{k+1}\left(\lambda_{m}\tau^{k}\left(\frac{3}{2}(\dot{\Lambda}\tau+\Lambda(k+1))\right)+\left\Vert A_{k}(\tau)\right\Vert \right)d\tau
\end{align}
where we have used the triangle inequality an the fact that $\Lambda$ and $\dot \Lambda$ are positive. Writing $(\dot{\Lambda}\tau+\Lambda(k+1))\tau^k=\frac{d}{d\tau}(\Lambda\tau^{k+1})d\tau$, integrating by parts and dropping the negative term leads to

\begin{align}
   \|E\| &\leq t^{k+1}\sum_{m=k+1}^{\infty}t^{m}\lambda_{m}+\frac{3}{2}t^{2k+2}\Lambda\sum_{m=0}^{k}\lambda_{m}^{2}t^{m}+\int_{0}^{t}\sum_{m=0}^{k}\tau^{m}\lambda_{m}\tau^{k+1}\left\Vert A_{k}(\tau)\right\Vert d\tau\\
   &\leq\frac{3}{2}t^{2k+2}\Lambda\sum_{m=0}^{k}\lambda_{m}^{2}t^{m}+t^{k+1}\left(\sum_{m=0}^{k}t^{m}\lambda_{m}\int_{0}^{t}\left\Vert A_{k}(\tau)\right\Vert d\tau+\sum_{m=k+1}^{\infty}t^{m}\lambda_{m}\right)\quad\mbox{using $t\geq\tau$}\\\nonumber
   &\leq t^{2k+2}\left(\frac{3}{2}\Lambda\sum_{m=0}^{k}\lambda_{m}^{2}t^{m}+\left(\sum_{m=0}^{\infty}t^{m}\lambda_{m}\right)^2+\sum_{m=k+1}^{\infty}t^{m-(k+1)}\lambda_{m}\right)\quad\mbox{using that $\int_{0}^{t}\|A_{k}(s)\|ds\leq t^{k+1}\sum_{m=0}^{\infty}t^{m}\lambda_{m}$}\\
   &\leq t^{2k+2}\left(\frac{3}{2}(k+1)|\lambda|_2^2+|\lambda|_{1}+|\lambda|_{1}^{2}\right)\quad\mbox{valid for $0\leq t\leq 1$, }
\end{align}
where $|\lambda|_p=\left(\sum_{m=0}^\infty\lambda_m^p\right)^{1/p}$ is the $\ell_p$ norm of $\lambda$.
\end{proof}

\subsection{Concentration properties}
The improved error scaling of \FAST{}  also implies improved concentration properties of {\it each} realization of the formula almost surely. This property is not unique to our method, and has been shown to exist for qDrift (although with a different scaling) and for other random formulas \cite{Chen2021,Kiss2023}. In our context, this is showcased in the following results
\begin{theorem}(Concentration bound of \FAST{})\label{theo:concentration}
Given the stochastic channel $\mathcal{Z}_{2k+1}$ on $n$ qubits and $V_{m, j}$ the $j$-th unitary appearing in the product of $N$ Trotter steps in the $m-$th observation $m\in(1,\dots M)$ on an instance of the sampled channel $\mathcal{Z}_{2k+1}$. 

The sample mean converges exponentially fast (concentrates) to the expected value as
\begin{align}
    {\rm Pr}\left(\left\|\frac{1}{M}\sum_{m=1}^M\prod_{j=1}^NS_kV_{m,j}-\mathbb{E}[S_kV]^N\right\|\geq \frac{\epsilon}{2}\right)\leq 2^{n+1}\exp\left(-\frac{M\epsilon^2}{40\tilde\lambda^2 N (k+1)^2}\left(\frac{N}{T}\right)^{2k+2}\right)
\end{align}
for a time step $T/N\leq 1$ and an error $\epsilon\leq 3(k+1)N \tilde\lambda(t/N)^{k+1}$. 
Here $\tilde\lambda = \lVert \boldsymbol{\alpha} \rVert_{\infty} \coloneqq \max_m\sum_r|\alpha_{mr}|$.
In order to guarantee an approximation error $\epsilon/2$ with probability $1-\delta$, it is sufficient to have
\begin{align}
    M\geq \frac{40\tilde\lambda^2(k+1)^2 N}{\epsilon^2}\left(\frac{T}{N}\right)^{2k+2}\log\left(\frac{2^{n+1}}{\delta}\right)
\end{align}
\end{theorem}

\begin{corollary}\label{coro:fluct}(Fluctuation bound of \FAST{})
Let $(\mathcal{Z}_{2k+1}^N)_{\exp}(\rho)$ be the sample averaged channel $(\mathcal{Z}_{2k+1}^N)_{\exp}(\rho):=\frac{1}{M}\sum_{m=1}^M\prod_{j=1}^NS_kV_{m, j}\rho V^{\dagger}_{m, j} S_k^\dagger$, then for a total evolution time $T$ the expected error between a sampled averaged channel and the true evolution satisfies
    \begin{align}
    \mathbb{E}\left[\frac{1}{2}\left\|(\mathcal{Z}_{2k+1}^N)_{\exp}(\rho)-\mathcal{U}^N\mathcal{}\right\|_\diamond\right]&\leq CN\left(\frac{T}{N}\right)^{2k+2}+\int_0^\infty{\rm Pr}\left(\left\|\frac{1}{M}\sum_{m=1}^M\prod_{j=1}^NS_kV_{m,j}-\mathbb{E}[S_kV]^N\right\|\geq \tau\right)d\tau\\
    &\leq CN\left(\frac{T}{N}\right)^{2k+2}+C_2\frac{2(k+1)\tilde\lambda}{3M}\left(\frac{T}{N}\right)^{k+1}\
    \end{align}
    where $C$ is the constant appearing in \cref{thm:approx_rand} and $C_2 := (\log(2^{n+1})+1)(1+\sqrt{1+6x})+\frac{3x}{\sqrt{1+6x}} $ with $x=\frac{MN}{\log(2^{n+1})}$. For $NM\gg\log(2^{n+1})\gg 1$
    \begin{align}\label{eq:scaling_total}
        \mathbb{E}\left[\frac{1}{2}\left\|(\mathcal{Z}_{2k+1}^N)_{\exp}(\rho)-\mathcal{U}^N\mathcal{}\right\|_\diamond\right]\leq CN\left(\frac{T}{N}\right)^{2k+2}+\alpha\frac{2(k+1)\tilde\lambda}{3M}\left(\frac{T}{N}\right)^{k+1}\sqrt{MN\log(2^{n+1})}
    \end{align}
    with $\alpha$ a numerical constant.
\end{corollary}

Let's unpack these results. \cref{eq:scaling_total} shows the scaling of the expected error between the mean of sampled channels and the exact evolution, for some time evolution time $t$ with $N$ Trotter steps. The first term comes from the bound between the stochastic channel and the exact evolution, while the second comes from a fluctuation bound that depends on the number of times that the channel is applied. In a computational model where the result appears after a single application of the time evolution channel, the first source of error is the only one that contributes to the error analysis. In the study of concentration properties, we want to understand the role of the number of measurements, which makes sense in a computational model where the solution to a problem is obtained by averaging several results obtained from a quantum computer. An example of the first computational model is quantum phase estimation, while the second is exemplified by statistical phase estimation. 

The bound \cref{eq:scaling_total} shows that depending on the value of the time step and the number of measurements, the error is dominated by the first or the second term, where the transition in scaling depends on the constants of the particular problem. An sketch of this behavior is presented in \cref{fig:scaling_sketch}. The same analysis but on concrete models is presented in \cref{sec:numerics}, where both scaling trends are shown against the data for better comparison.

Fixing the error $\epsilon$, and the evolution time $T$ the gate complexity scales as
\begin{align}
    N = O\left(T\left(\frac{T}{\epsilon^{2}}\right)^{\frac{1}{2k+1}}\right)\max\left(O\left(\frac{1}{M^{\frac{1}{2k+1}}}\right),O\left(\epsilon^{\frac{1}{4k+2}}\right)\right)
\end{align}
matching the scaling of a $2k$-order product formula. We would like to emphasize that this result only requires a product formula of order $k$ and the implementation of an stochastic unitary whose depth is at most logarithmic in the number of qubits. We discuss the circuit construction in practice in \cref{sec:depth_diss}.
We present numerical results showing this scaling for some particular spin and fermionic systems relevant for many-body and molecular studies in \cref{sec:numerics}.

Proving \cref{theo:concentration}  and \cref{coro:fluct} is the main focus of  the rest of this section. To do so, it is useful to recall some concepts from probability theory. 
\begin{definition}[$\sigma$-algebra]
    A $\sigma$-algebra $\mathcal{F}_j$ is a collection of subsets of a set $\Omega$ that
    \begin{enumerate}
        \item Contains the whole space, $\Omega \in \mathcal{F}_j$,
        \item is closed under complements,  i.e. f $A\in \mathcal{F}_j$ then $A^c\in\mathcal {F}_j$ as well
        \item is closed under countable unions, such that if $\{A_r\}_{r=1}^\infty \in \mathcal{F}_j$ then also $\cup_{r=1}^{\infty}A_r\in \mathcal{F}_j$.
    \end{enumerate} 
\end{definition} 
In simple terms a $\sigma$-algebra contains all events to which probabilities can be assigned at at a particular time. To discuss the events as time unfolds, 
we use \cref{def:martingale} and \cref{fact:freedman} introduced in \cite{Chen2021} and \cite{Kiss2023}, adapted to our algorithm.

\begin{definition}(Martingale)\label{def:martingale}
Consider a master $\sigma$-algebra $\mathcal{F}=\cup_j \mathcal{F}_j$ and a filtration of $\mathcal{F}$, $\{\mathcal{F}_t\}_{t\geq 0}$ such that $\mathcal{F}_s\subseteq \mathcal{F}_t$ for $s\leq t$. A martingale is a sequence of random variables $\{B_1,B_2,\dots,\}$ satisfying
\begin{enumerate}
    \item\label{causality} The $\sigma$-algebra generated by $B_k$ is a subset of $\mathcal{F}_k$, i.e. $\sigma(B_k)\subset \mathcal{F}_k$,
    \item\label{status_quo} $\mathbb{E}[B_{k+1}|B_kB_{k-1}\dots B_0]=B_{k}$.
\end{enumerate}
\end{definition}
Plainly speaking, a filtration gives a structured way to model the progressive revelation of information over time.
The master $\sigma$-algebra represents the total accumulated information in the system, encompassing all past, present, and future knowledge. Condition \ref{causality} above represents causality, while condition \ref{status_quo} represents that the present is on average, the same as the past.
The following result appears in \cite{Chen2021} as a simplification of the more general result \cite{Tropp2011}.

\begin{fact}(Simplified matrix Freedman inequality)\label{fact:freedman}
    Let $\mathbb{M}_{d\times d}$ be the space of $d\times d$ matrices and $\{B_k: k = 0, . . . , N \}\subset \mathbb{M}_{d\times d}$ be a
matrix martingale . 
Assume that the sequence $C_k:= B_k - B_{k-1}$ obeys
$\|C_k\|\leq R$ and its conditional variance obeys $\|\sum_{k=1}^N \mathbb{E}[C_kC_k^\dagger|C_{k-1}\dots C_1] \|\leq v$ almost surely. Then
\begin{align}
    {\rm Pr}(\| B_N - B_0\|\geq \tau)\leq 2d\exp\left(\frac{-3\tau^2}{6v+ 2R\tau}\right)\quad\mbox{for any positive $\tau$}.
\end{align}
\end{fact}

We want to analyse the concentration properties of \FAST{} as we increase the number of Trotter steps $N$, and the numbers of observations (samples) $M$. Here we largely follow the analysis in \cite{Kiss2023}, with the modifications needed in our scenario. We define the martingales
\begin{align}
    (B^{(k)})_p^m:=\mathbb{E}[S_k V]^{N-p}\prod_{r=1}^p(S_kV_{m, r})
\end{align}
with $m\in(1,\dots, M)$ indexing the observations, and $p\in(0,1,\dots, N)$ the Trotter step (including $p=0$, no Trotter step). Here the unitary $V_{m, r}$ is the $r$-th unitary appearing in the $m$-th sample. $S_k$ is a $k$-order product formula.
The martingale $(B^{(k)})_p^m$ satisfies the causality condition by construction as it only depends on the time steps previous to $p$. Condition \ref{status_quo} is satisfied directly as
\begin{align}
    \mathbb{E}[(B^{(k)})_p^m|(B^{(k)})_{p-1}^m\cdots(B^{(k)})_1^m(B^{(k)})^m_0]=\mathbb{E}[S_kV]^{N-p}\mathbb{E}[S_kV_{m, p}]\prod_{r=1}^{p-1}(S_kV_{m, r})=(B^{(k)})^m_{p-1}.
\end{align} 
In order to define an interpolating martingale we stack up the observations (see \cref{fig:obs}), such that the for the sampled step $i\in (0,1,\dots, N)$ of the observation $m\in (1,\dots M)$ we introduce the index $I(i,m):=i+(m-1)N$ and the random variable 
\begin{align}
    D^{(k)}_{I(i,m)}=D^{(k)}_{i+(m-1)N}:&=\sum_{r=m+1}^M (B^{(k)})_0^r + (B^{(k)})_i^m + \sum_{r=1}^{m-1}(B^{(k)})^r_N
    \label{eq:compositemg}
\end{align}
For $I\in(0,\dots,MN)$, $D^{(k)}_I$ defines a martingale, as the sum of martingales with respect to the same filtration is also a martingale by linearity of the expectation. 
Note that $I(N, m) = I(0, m + 1)$ for all $m < M$, and $D^{(k)}_{I(N, m)} = D^{(k)}_{I(0, m + 1)}$ as well. 
That is, the bottom element of each column in \cref{fig:obs} is identified with the top element of the adjacent column to the right.
Now we define the sequence
\begin{align}
    C^{(k)}_{J(j,m)}&:=D^{(k)}_{J(j,m)}-D^{(k)}_{J(j,m)-1}=D^{(k)}_{j+(m-1)N}-D^{(k)}_{j-1+(m-1)N}=(B^{(k)})_{j}^m - (B^{(k)})_{j-1}^m\\
    &=\mathbb{E}[S_kV]^{N-j}(S_kV_{m, j}-\mathbb{E}[S_kV])\prod_{r=1}^{j-1}(S_kV_{m, r})
\end{align}
This form allows us to bound $\|C_J\|$ as $\|S_kV\|\leq 1$ and $\|\mathbb{E}(S_kV)\|\leq 1$, so
\begin{align}
\|C_J\|&\leq\|V_{m, j}-\mathbb{E}(V)\|\leq \|V_{m, j}-1\|+\|1-\mathbb{E}(V)\|\quad\mbox{by the triangle inequality},\\
&\leq \|V_{m, j}-1\|+ \mathbb{E}(\|1-V\|) \quad\mbox{by Jensen's inequality},\\
&\leq 2\max_m|\theta_{mr}|=2\left(\frac{t}{N}\right)^{k+1}\Lambda(t/N)\max_m\lambda_m,
\end{align}
where in the last line we have used that $\|e^{ix}-1\|<\|x\|$ for any hermitian matrix $x$, and $V\in\{\{e^{i\theta_{mr}(t)P_r}\}_{r=1}^{r_{\rm max}}\}_{m=0}^k$ with $\theta_{mr}=\frac{t^{k+1}}{N^{k+1}}{\rm sign}(\alpha_{mr})\Lambda(t/N)\lambda_m$ as per \cref{def_U}. We can simplify this bound assuming a time step $t/N\leq 1$, and defining $\tilde{\lambda}:=\max_m\lambda_m$ which leads to the bound
\begin{align}\label{eq:bound_C}
    \|C_J\|\leq 2(k+1)\tilde{\lambda}\left(\frac{t}{N}\right)^{k+1}\quad\mbox{valid for $\frac{t}{N}\leq  1$}.
\end{align}

\begin{figure}
    \centering
    \includegraphics[width=0.8\linewidth]{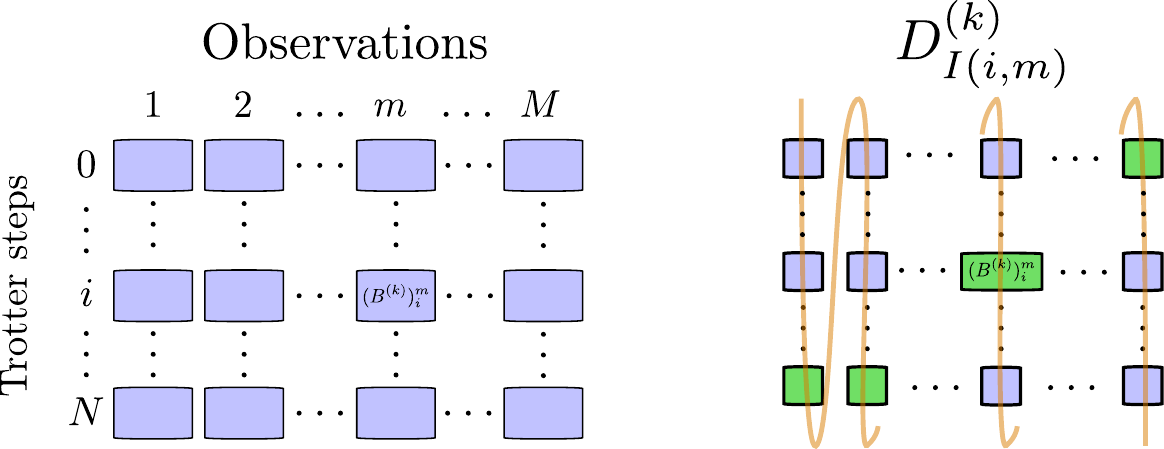}
    \caption{Left. The $M$ observations define a set of $M$ martingales $\{\{(B^{(k)})_i^m\}_{i=0}^N\}_{m=1}^M$. Right. The martingale $D^{(k)}_{I(i,m)}$ is defined by the sum over the green martingales, where the ordering is defined through the zig-zag line.
    By our indexing scheme, we identify $D^{(k)}_{I(N, m)} = D^{(k)}_{I(0, m + 1)}$ for all $m < M$.
    }
    \label{fig:obs}
\end{figure}
We can also bound the conditional variance as
\begin{align}\label{eq:bound_C2}
    \left\|\sum_{J=1}^{MN}\mathbb{E}[C_JC_J^\dagger|C_{J-1}\dots C_1]\right\|\leq MN{\rm max}_J\|C_J\|^2=4(k+1)^2MN\left(\frac{t}{N}\right)^{2k+2}\tilde{\lambda}^2
\end{align}
We can now apply directly \cref{fact:freedman} to prove \cref{theo:concentration} 
\begin{proof}(\cref{theo:concentration})
\begin{align}
         {\rm Pr}\left(\left\|\frac{1}{M}\sum_{m=1}^M\prod_{j=1}^NS_kV_{m,j}-\mathbb{E}[S_kV]^N\right\|\geq \frac{\epsilon}{2}\right)&={\rm Pr}\left(\left\|D^{(k)}_{MN}-D_0^{(k)}\right\|\geq \frac{M\epsilon}{2}\right)\\
         &\leq 2^{n+1}\exp\left(-\frac{3M\epsilon^2}{8(k+1)\tilde\lambda}\left(\frac{N}{t}\right)^{k+1}\frac{1}{12N (k+1)\frac{t^{k+1}}{N^{k+1}}\tilde{\lambda}+\epsilon}\right)
\end{align}
where we have used the definition of the martingale $D^{(k)}_J$ and \cref{fact:freedman} with \cref{eq:bound_C} and \cref{eq:bound_C2}, assuming $t/N\leq 1$. We can further simplify this bound if we take that the error $\epsilon\leq 3N (k+1)\tilde\lambda(t/N)^{k+1}$, from where it follows that
\begin{align}
    {\rm Pr}\left(\left\|\frac{1}{M}\sum_{m=1}^M\prod_{j=1}^NS_kV_{m,j}-\mathbb{E}[S_kV]^N\right\|\geq \frac{\epsilon}{2}\right)\leq 2^{n+1}\exp\left(-\frac{M\epsilon^2}{40\tilde\lambda^2N(k+1)^2}\left(\frac{N}{t}\right)^{2k+2}\right)
\end{align}
\end{proof}
The proof of \cref{coro:fluct} is analogous to the one appearing in \cite{Chen2021,Kiss2023}, with the bounds \cref{eq:bound_C} and \cref{eq:bound_C2}, and we reproduce it here for completeness
\begin{proof}(\cref{coro:fluct})

Using the triangle inequality we have 
\begin{align}
    \mathbb{E}\left[\frac{1}{2}\left\|(\mathcal{Z}_{2k+1}^N)_{\exp}(\rho)-\mathcal{U}^N\mathcal{}\right\|_\diamond\right]\leq \frac{1}{2}\left\|(\mathcal{Z}_{2k+1}^N)(\rho)-\mathcal{U}^N\mathcal{}\right\|_\diamond+\mathbb{E}\left[\frac{1}{2}\left\|(\mathcal{Z}_{2k+1}^N)_{\exp}(\rho)-(\mathcal{Z}_{2k+1}^N)(\rho)\mathcal{}\right\|_\diamond\right].
\end{align} The first term is bounded in \cref{thm:approx_rand} and we can bound the second using \cref{theo:concentration} as
    \begin{align}\nonumber
    \mathbb{E}\left[\frac{1}{2}\left\|(\mathcal{Z}_{2k+1}^N)_{\exp}(\rho)-(\mathcal{Z}_{2k+1}^N)(\rho)\mathcal{}\right\|_\diamond\right]&=\int_0^\infty{\rm Pr}\left(\left\|\frac{1}{M}\sum_{m=1}^M\prod_{j=1}^NS_kV_{m,j}-\mathbb{E}[S_kV]^N\right\|\geq \tau\right)d\tau\\
    &\leq \int_0^\infty{\rm min}\left(1,2^{n+1}\exp\left(-\frac{3M\tau^2}{4R}\frac{1}{6NR +\tau}\right)\right)d\tau\\\label{eq:split}
    &=\int_0^{\tau_+}d\tau +2^{n+1}\int_{\tau_{+}}^\infty \exp\left(-\frac{3M\tau^2}{4R}\frac{1}{6NR +\tau}\right)d\tau
    \end{align}
    with $R:= (t/N)^{k+1}\tilde\lambda (k+1)$  and $\tau_\pm:=\frac{2\log(2^{n+1})R}{3M}\left(1 \pm \sqrt{1+\frac{6NM}{\log(2^{n+1})}}\right)$ is positive/negative solution to the equation $ \frac{3}{4}M\tau^2/R-(6NR + \tau)\log(2^{n+1})=0$. The first term in \cref{eq:split} is just $\tau_+$, while the second can be bounded as
    \begin{align}
        2^{n+1}\int_{\tau_+}^\infty\exp\left(-\frac{3M\tau^{2}}{4R(6RN+\tau)}\right)d\tau&=\int_{\tau_+}^\infty\exp\left(-\frac{3M(\tau-\tau_{+})(\tau-\tau_{-})}{4R(6RN+\tau)}\right)d\tau\\
        &\leq \int_{\tau_+}^\infty\exp\left(-\frac{3M(\tau-\tau_{+})}{4R}\frac{\tau_+-\tau_{-}}{6RN+\tau_+}\right)d\tau=\frac{2R}{3M}\left(\frac{9x+1}{\sqrt{1+6x}}+1\right)
    \end{align}
    where  $x=\frac{NM}{\log(2^{n+1})}$. Collecting everything we have
    \begin{align}
        \mathbb{E}\left[\frac{1}{2}\left\|(\mathcal{Z}_{2k+1}^N)_{\exp}(\rho)-(\mathcal{Z}_{2k+1}^N)(\rho)\mathcal{}\right\|_\diamond\right]&\leq\frac{2(k+1)\tilde\lambda}{3M}\left(\frac{t}{N}\right)^{k+1}\left[(\log(2^{n+1})+1)(1+\sqrt{1+6x})+\frac{3x}{\sqrt{1+6x}}\right]
    \end{align}
    which proves the theorem.
\end{proof}

\section{Generalizations}\label{sec:generalizations}

\subsection{Other sampling strategies}\label{sec:sampling}

In the form of \FAST{} described in Algorithm~\ref{algo:gen} only one Pauli term is sampled, however we can also sample a commuting set from the $\Omega_j$. To separate by commuting sets, let's introduce some notation. 
Let $\Xi$ denote a set of mutually pairwise commuting terms, and $\xi\in \Xi$ a Pauli within a commuting set. Let's denote the set of commuting sets $C$. We can approximate the unitary $\mathcal{T}e^{i\int_{0}^{t}A_{k}(s)ds}$ to order $O(t^{2k+2})$ by
\begin{align}
    \mathcal{T}e^{i\int_{0}^{t}A_{k}(s)ds}=\sum_{m=0}^{k}\frac{t^{m}}{\Lambda(t)}\sum_{\Xi\in C}\frac{\sum_{\xi\in \Xi}|\alpha_{m\xi}|}{\sum_{\Xi'\in C}\sum_{\xi'\in \Xi'}|\alpha_{m\xi'}|}e^{it^{k+1}\sum_{\xi\in \Xi}\frac{\alpha_{m\xi}}{\sum_{\xi'\in \Xi}|\alpha_{m\xi'}|}\sum_{\Xi'\in C}\sum_{x\in \Xi'}|\alpha_{mx}|\Lambda(t)P_{\xi}}+O(t^{2k+2})
\end{align} 
note that $\sum_{\Xi\in C}\sum_{x\in \Xi}|\alpha_{mr}|=\sum_{r}|\alpha_{mr}|=\lambda_m$ so we have 
\begin{align}
    \mathcal{T}e^{i\int_{0}^{t}A_{k}(s)ds}=\sum_{m=0}^{k}\frac{t^{m}}{\Lambda(t)}\sum_{S\in C}\frac{\sum_{s\in S}|\alpha_{ms}|}{\sum_{r}|\alpha_{mr}|}e^{i\sum_{r\in S}t^{k+1}\frac{\alpha_{mr}}{\sum_{s\in S}|\alpha_{ms}|}\lambda_m\Lambda(t)P_{r}}+O(t^{2k+2}).
\end{align}
For each individual Pauli $\alpha_{mr} P_r$ within the sampled $\Xi$, the $\theta_{mr}$ which we rotate by are then adjusted accordingly to
\begin{equation}
    \theta_{mr} = \frac{\alpha_{mr} \sum_p{|\alpha_{mp}|}}{\sum_{\xi \in \Xi}{|\alpha_{m\xi}|}} \left(\sum_{s=k}^{2k} \frac{t^{s+1}}{s+1}\right)
\end{equation}
where $k$ is the Trotter order.

The set of commuting sets $C$ can be selected in a variety of ways. If the aim is to reduce the gate depth of the circuit implemented on hardware, we could select terms that act on disjoint qubits so that all the terms in each set $\Xi$ are simultaneously implementable. Additional hardware considerations such as qubit connectivity could also be considered to determine whether terms are parallelisable. A middle ground between fully commuting sets and qubit disjoint sets could also be chosen, for example by grouping together a Pauli term that requires a larger gate depth with a selection of shorter terms.

\subsection{Symmetric \FAST{}}	
\label{sec:general_ordering}

The protocol discussed in \cref{sec:analysis} assumes that the unitary characterizing the error is multiplied at the end of the product formula. This can be generalized such that, given a product formula, the error is a unitary inserted somewhere in the product. Concretely, let's take a product formula $S_k(t)$ and a partition $S_k^L(t)$, $S_k^R(t)$ such that $S_k^L(t)S_k^R(t)=S_k(t)$. We can define the error unitary such that
\begin{align}
    U=e^{itH}=S_k^L(t)F(t)S_k^R(t).
\end{align}

Following the same steps as in \cref{sec:analysis}, we find that $F(t)$ satisfies the equation
\begin{align}\label{eq:err_sym}
    \frac{dF}{dt}=i\left(-i\frac{dS_{k}^{L\dagger}}{dt}S_{k}^{L}+\frac{1}{2}S_{k}^{L\dagger}HS_{k}^{L}\right)F+iF\left(-iS_{k}^{R}\frac{dS_{k}^{R\dagger}}{dt}+\frac{1}{2}S_{k}^{R}HS_{k}^{R\dagger}\right).
\end{align}
Here both operators appearing in the round brackets are hermitian (see \cref{eq:hermitian}) so we define
\begin{align}\label{eq:definitions}
    A^L_k(t):=-i\frac{dS_{k}^{L\dagger}}{dt}S_{k}^{L}+\frac{1}{2}S_{k}^{L\dagger}HS_{k}^{L},\\
    A_k^R(-t):=-iS_{k}^{R}\frac{dS_{k}^{R\dagger}}{dt}+\frac{1}{2}S_{k}^{R}HS_{k}^{R\dagger},
\end{align}
using this \cref{eq:err_sym} is solved by (see \cite{Dollard_Friedman_1984} p.12, or appendix A in \cite{Childs2021})
\begin{align}\label{eq:sol_sym}
    F(t)=\mathcal{T}e^{i\int_{0}^{t}A_k^L(s)ds}\mathcal{T}e^{i\int_{-t}^{0}A_k^R(s)ds}
\end{align}
We know that by virtue of $S_k$ being a $k-$th order formula, $F$ should be the identity up to a correction of order $k+1$. From the integral representation of the error we have on the other hand
\begin{align}
    F(t)-1=\int_0^t\frac{dF}{ds}ds=i\int_0^t\left(iA_k^L(s)F(s)+iF(s)A_k^R(-s)\right)ds
\end{align}
from where it follows, using that $F(t)-1=O(t^{k+1})$ 
\begin{align}
    \int_0^t\left(A_k^L(s)+A^R_k(-s)\right)ds= O(t^{k+1})
\end{align}
To apply the rest of the \FAST{} protocol it would be useful to express the error $F$ in terms of a single exponential, from where a sample procedure can be implemented. Note that applying a first order formula to consolidate the exponents in \cref{eq:sol_sym} is not enough as this would incur in an error proportional to the commutator $\int_0^t\int_0^s||[A_k^R(\tau),A^L_k(s)]||d\tau ds$, \cite{J_Huyghebaert_1990}, which could be larger than $O(t^{k+1})$.
Instead we define the effective time dependent Hamiltonian $\mathcal{A}_k$ such that
\begin{align}\label{eq:gen_ham}
    \mathcal{A}_k(t):=A_k^L(t)+FA_k^R(-t)F^\dagger
\end{align}
and $F=\mathcal{T}e^{i\int_0^t\mathcal{A}_k(s)ds}$. The few first explicit terms of this effective Hamiltonian are
\begin{align}
    \mathcal{A}_k(t)=A_{k}^{L}(t)+A_{k}^{R}(-t)+i\int_{0}^{t}[A_{k}^{L}(s)+A_{k}^{R}(-s),A_{k}^{R}(-t)]ds+\dots
\end{align}
putting all of these ingredients together we can articulate the general \FAST{} algorithm

\begin{algorithm}[H]
\SetAlgoLined
\textbf{Input:} Given a Hamiltonian $H=\sum_{j=1}^l h_j$, where $h_j$ are operators such that $e^{ith_j}$ is directly implementable, a $k$-th order product formula $S_k^{(l)}(t)$, a evolution time $t$, a splitting $S^{(l)L}_k(t)$, $S^{(l)R}_k(t)$ such that $S^{(l)L}_k(t)S^{(l)R}_k(t)=S^{(l)}_k(t)$.\\
\KwResult{Stochastic approximant of $e^{itH}$ with error of order $O(t^{2k+2})$, general version}

\medskip

 Compute the series expansion in time of \cref{eq:gen_ham} up to order $t^{2k+1}$, i.e
\begin{align}
\mathcal{A}_k^{(l)}(t)=\sum_{m=k}^{2k}t^m\Omega_m +O(t^{2k+1})
\end{align}
Store the list of powers with non-zero coefficient in the above expansion as $\mathbf{m}$ 

\For{$m\in\mathbf{m}$}{
Compute $\Omega_m$ as a sum of Pauli terms $\Omega_m=\sum_s\alpha_{ms}P_s$\\
Output the probability distribution $p_{ms}:=\frac{|\alpha_{ms}|}{\sum_{p}|\alpha_{mp}|}$
} 
Define the probability distribution $p_m(t):=\frac{t^{m+1}}{m+1}\left(\sum_{s\in \mathbf{m}}\frac{t^{s+1}}{s+1}\right)^{-1}$ for $m\in\mathbf{m}$

 Sample the random variable $M\in\mathbf{m}$ from the probability distribution $p_m(t)$\\
 Sample the random variable $R$ from the probability distribution $p_{m|r}$, i.e. ${\rm Pr}(R=r|M = m)=p_{m|r}$.\\
 Construct the estimator
 $X(t)=S_k^{(l)L}(t)e^{i\theta_{mr}(t)P_r}S_k^{(l)R}(t)$
 where $\theta_{mr}(t)={\rm sign}(\alpha_{mr})\left(\sum_{s\in \mathbf{m}}\frac{t^{s+1}}{s+1}\right)\sum_{p}|\alpha_{mp}|$ 
 Output $X(t)$
 \caption{Approximating time evolution with stochastic formulas}\label{algo:gen_sym}
\end{algorithm}

A particularly important set of formulas corresponds to $S^{L}_k(t)=V_k(t)$ and $S^R_k(t)=V_k^\dagger(-t)$, where $V_k(t)$ is on itself a $k$-th order product formula. We call this set of formulas symmetric \FAST{}.

As an example, let's consider $V_1=e^{\frac{it A}{2} } \, e^{\frac{it B}{2}  } $ a first order product formula such that
\begin{align}
    e^{it(A+B)}=e^{\frac{it A}{2} } e^{\frac{it B}{2}  } F(t) e^{\frac{itB}{2} } e^{\frac{itA}{2} }. 
\label{eq:symzh}
\end{align}
Using \cref{eq:definitions} we have
\begin{align}
    A_1^L(t)&=\frac{1}{2}e^{-i\frac{t}{2}B}(e^{-i\frac{t}{2}A}Be^{i\frac{t}{2}A}-B)e^{i\frac{t}{2}B}:=A_1(t)\\
    A_1^R(-t)&=A_1(-t)
\end{align}
a direct computation reveals that
\begin{align}
    \mathcal{A}_1(t)=-6t^{2}C_{3}+10t^{4}C'_{5}+\frac{t^{4}}{2}\left[C_{3},{\rm Ad}_{A}(B)\right]+O(t^{6})
\end{align}
where
\begin{align}
    C_3&=\frac{1}{24}\left(\frac{1}{2}{\rm Ad}_{A}^{2}(B)+{\rm Ad}_{B}{\rm Ad}_{A}(B)\right)\\
    C'_5&=\frac{1}{160}\left(\frac{1}{24}{\rm Ad}_{A}^{4}(B)+\frac{1}{6}{\rm Ad}_{B}{\rm Ad}_{A}^{3}(B)+\frac{1}{4}{\rm Ad}_{B}^{2}{\rm Ad}_{A}^{2}(B)+\frac{1}{6}{\rm Ad}_{B}^{3}{\rm Ad}_{A}(B)\right)
\end{align}
Integrating we find
\begin{align}
    \int_{0}^{t}\mathcal{A}_{k}(s)ds&=-2t^{3}C_{3}+2t^{5}\left(C'_{5}+\frac{1}{20}\left[C_{3},{\rm Ad}_{A}(B)\right]\right)+O(t^7)\\
    &=-2t^{3}C_{3}+2t^{5}C_5+O(t^7)
\end{align}
Thus we can write the symmetric formula
\begin{align}
    e^{it(A+B)}=e^{\frac{it A}{2} } e^{\frac{it B}{2}  }e^{-it^3C_3}e^{it^5C_5}e^{it^5C_5}e^{-it^3C_3} e^{\frac{itB}{2} } e^{\frac{itA}{2} }+O(t^7)
\end{align}
Comparing with Eq. 2.20 in \cite{Arnal2017} , we see that we recover the symmetric Zassenhaus formula.

\subsection{Generalized Zassenhaus formulas}\label{sec:general_zas}

The general procedure that we use to generate \FAST{} can be used to define generalized Zassenhaus formulas, something that could be of independent interest. Here we discuss this relation, which leads to recursive formulas with the same structure as the one discussed in \cite{Casas_2012}, but with different initial conditions. We revisit these results in our context here.

The parametrization of the error in a product formula (note that here we do not assume that the exponential is unitary) has the following formal expansion
\begin{align}
    e^{t(X+Y)}	=S_{k}(t)\prod_{j=k+1}^{\infty}e^{t^{j}\Omega_{j}}
e^{t(X+Y)}	=S_{k}(t)\prod_{j=k+1}^{n}e^{t^{j}\Omega_{j}}R_{n}(t)
\end{align}
where $R_n$ is defined as
\begin{align}\label{eq:def_R2}
    R_{n}(t)=\prod_{j=n}^{k+1}e^{-t^{j}\Omega_{j}}S_{k}^{-1}(t)e^{t(X+Y)}
    :=e^{-t^{n}\Omega_{n}}R_{n-1}(t)\quad\mbox{for $n\geq k$}.
\end{align} From here it follows that the error satisfies 
\begin{align}\label{eq:def_R1}
    R_{n}(t)=\prod_{j=n+1}^{\infty}e^{t^{j}\Omega_{j}}
\end{align}. Defining $F_{n}:=\left(\frac{d}{dt}R_{n}\right)R_{n}^{-1}$ and differentiating \cref{eq:def_R2} we obtain
\begin{align}
    \frac{d}{dt}R_{n}	=\left[-nt^{n-1}\Omega_{n}+e^{-t^{n}\Omega_{n}}\frac{d}{dt}R_{n-1}(t)R_{n-1}^{-1}e^{t^{n}\Omega_{n}}\right]R_{n}\\\label{eq:def_F1}
    F_{n}	=-nt^{n-1}\Omega_{n}+e^{-t^{n}\Omega_{n}}F_{n-1}e^{t^{n}\Omega_{n}}
	=e^{t^{n}{\rm Adj}_{\Omega_{n}}}(-nt^{n-1}\Omega_{n}+F_{n-1})
\end{align}
    Now, differentiating \cref{eq:def_R1} we have
    \begin{align}
        F_{n}=(n+1)t^{n}\Omega_{n+1}+e^{t^{n+1}\Omega_{n+1}}\frac{d}{dt}\left(\prod_{j=n+2}^{\infty}e^{t^{j}\Omega_{j}}\right)\left(\prod_{j=\infty}^{n+2}e^{-t^{j}\Omega_{j}}\right)e^{-t^{n+1}\Omega_{n+1}}
    \end{align}
    finally, defining
    \begin{align}\label{eq:G_def}
    G_{n+1}(t):=e^{t^{n+1}\Omega_{n+1}}\frac{d}{dt}\left(\prod_{j=n+2}^{\infty}e^{t^{j}\Omega_{j}}\right)\left(\prod_{j=\infty}^{n+2}e^{-t^{j}\Omega_{j}}\right)e^{-t^{n+1}\Omega_{n+1}}
    \end{align}
    and using \cref{eq:G_def} for $G_{n}$ into \cref{eq:def_F1} we find the recurrence equations
    \begin{align}
        F_{n}(t)	=e^{t^{n}{\rm Adj}_{\Omega_{n}}}(G_{n}(t))\\
\Omega_{n}	=\frac{1}{n!}\frac{d^{n-1}}{dt^{n-1}}F_{n-1}(0)\\
G_{n}(t)	=F_{n-1}-nt^{n-1}\Omega_{n}.
    \end{align}
For $n=k$, we have the initial condition for the recurrence equations
\begin{align}
    F_{k}	=\left(\frac{d}{dt}R_{k}\right)R_{k}^{-1}=\left(\frac{d}{dt}S_{k}^{-1}(t)e^{t(X+Y)}\right)e^{-t(X+Y)}S_{k}
F_{k}	=\left(\frac{d}{dt}S_{k}^{-1}\right)S_{k}+S_{k}^{-1}(t)(X+Y)S_{k}.
\end{align}
which depends on the choice of the product formula $S_k$. Solving the recurrence relations generates a generalized Zassenhaus formula, where the initial exponential $e^{t(A+B)}$ is approximated by an infinite product of exponentials. Although it could be interesting to study the convergence radius of this approximation as a function of the product formula order $k$, we leave it as an exploration for the future. 

\section{Numerical results}
\label{sec:numerics}

In order to assess the performance of the \FAST{} algorithm, we carry out a range of simulations involving the Ising, Heisenberg and Hubbard models. 
In this section, we compare \FAST{} with the second- and fourth-order Trotter formulas~\cite{Susuki1991} which we label as ``Trotter 2'' and ``Trotter 4''. 
The main metric of interest is the operator norm of the difference between the exact and approximate time-evolved state $\|e^{itH}|\Psi\rangle-\bar{X}(t)|\Psi\rangle\|$. Here $e^{itH}$ is the exact time-evolution unitary and $\bar{X}(t) |\Psi\rangle$ is an average over statevectors $\bar{X}(t) |\Psi\rangle = \frac{1}{N_s} \sum_{j=1}^{N_s} X_j(t) |\Psi\rangle$ where $N_s$ is the number of \FAST{} samples and $X_j(t)$ are estimators generated via~\cref{algo:gen}. As this norm upper bounds the difference in diamond norm of the channels, we consider this quantity as the measure of error.

In addition to the standard implementation of \FAST{} defined in~\cref{algo:gen}, we make a number of modifications that we expect to improve the performance of the method in practice. 
The first modification we make is to introduce a \textit{greedy} version of the algorithm, where we independently sample a random Pauli term in each $\Omega_j, j\in\{k,\hdots\,2k\}$ instead of just one $\Omega_j$ chosen by the random variable $M$. Following the notation in \cref{algo:gen}, the estimator $X(t)$ becomes
\begin{equation}
    X(t) = S^{(l)}_k(t) \prod_{m=k}^{2k} e^{i\theta_{mr}(t)P_r}, \qquad \theta_{mr}(t) = \frac{t^{m+1}}{m+1} \text{sign}(\alpha_{mr}) \sum_p |\alpha_{mp}|.
\end{equation}
This approach is chosen to compensate for the fact that, for small $t$, terms in the high-order $\Omega_j$ are less likely to be selected. 

As an additional modification, we consider sampling commuting sets of Paulis, as described in \cref{sec:sampling}, rather than taking just one term in each sample. 
Here, we restrict to sampling qubit-disjoint sets (QDS), where each term in the set acts on a disjoint set of qubits. 
This is done so that the terms in each set are implementable simultaneously.
We use a graph partitioning algorithm to construct a collection of term-wise disjoint QDS and sample from this collection.
The QDS and greedy modifications can also be combined. 

Numerical simulations were performed using the Yao.jl quantum simulation library~\cite{yao} to run circuits and obtain exact time-evolved states. For speed, larger systems with deeper circuits were run on a GPU using cuQuantum~\cite{cuquantum} with the Julia interface provided by CUDA.jl~\cite{cuda_jl}.

\subsection{Transverse-field (TF) Ising model}

We begin our numerical simulations with the transverse-field Ising model on a 1D line and a 2D square lattice. The Hamiltonian is as follows:
\begin{equation} \label{eq:ising}
    H_\text{Ising}= J \sum_{\langle i, j \rangle} X_iX_j + h \sum_i Z_i
\end{equation}
where $X_i$ and $Z_i$ are the Pauli $X, Z$ operators on site $i$, $\langle i, j \rangle$ represents neighboring sites on the lattice, $J$ is the interaction strength and $h$ is the magnetic field strength. For the simulations in this section we take $J = h = 1$ and choose random computational basis initial states to evolve.

\subsubsection{1D TF Ising model}

In 1D this model can be mapped into an exactly solvable fermion system. We use this structure to reach larger systems sizes than what would be normally possible by evolving the system classically, and assess the quality of the \FAST{} algorithm at scale. We highlight that the structure of \FAST{} does not use the solvability of the model to improve the performance of product formulas, and as such we do not expect thfor this, which are further described in section Section 6. Note the 10x reduction in circuit resources commpared toe performance on this example to be particularly different than in a general case. To complement this analysis and provide further evidence about the performance of \FAST{} on generic (i.e non integrable) systems, we also consider other examples as well and we observe that the general trends are independent of the model studied. 

\begin{figure}[h!]
    \centering
    \includegraphics[width=\textwidth]{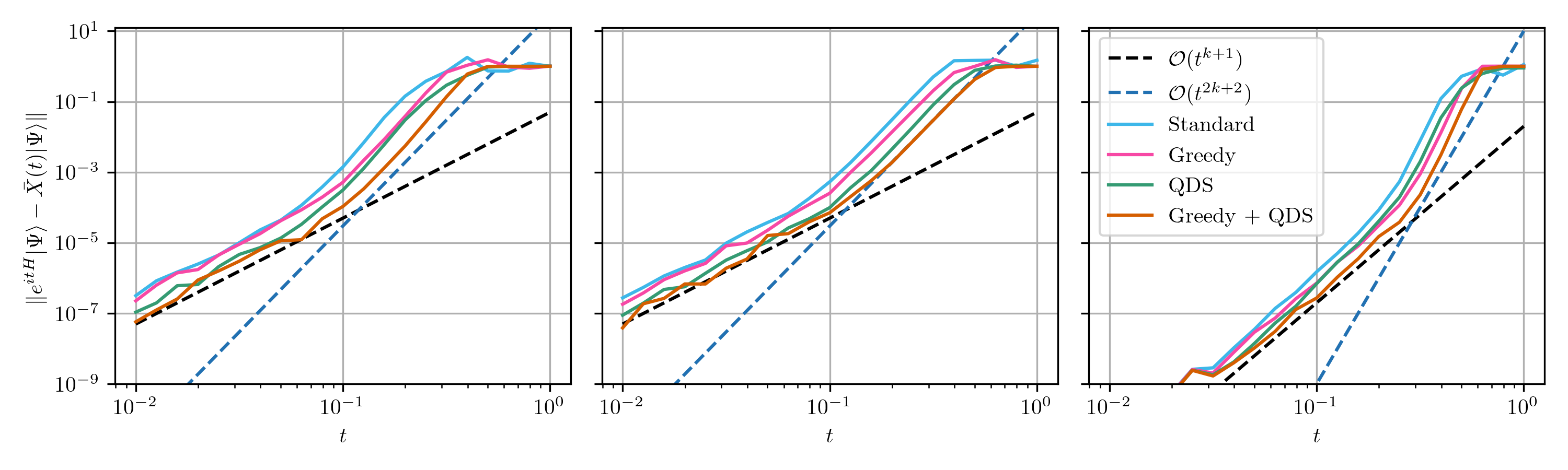} 
    \caption{Approximation error for the $1 \times 16$ TF Ising model with all \FAST{} variants taking 10,000 samples. From left to right: \FAST{} with $k=2$, symmetric \FAST{}, and \FAST{} with $k=4$. The $\mathcal{O}(t^{k+1})$ and $\mathcal{O}(t^{2k+2})$ lines have been added to aid the eye and to demonstrate the behaviour sketched in \cref{fig:scaling_sketch}.}
    \label{fig:ising_1x16_variants}
\end{figure}

In \cref{fig:ising_1x16_variants} we explore the error of some variants of \FAST{} with respect to the exact evolution. We highlight the two scaling behaviours that appear in these results. For sufficiently large $t$, the scaling of the error is $O(t^{2k+2})$. This scaling becomes $O(t^{k+1})$ as $t$ decreases. This behaviour is exactly predicted by our \cref{coro:fluct}, where the error is controlled by two contributions, the one coming from the bound on diamond norm between the stochastic channel and the exact unitary, which scales like $O(t^{2k+2})$ and the fluctuation bound which scales like $O(t^{k+1})$, but with a prefactor that depends on the number of measurements. For small enough time steps, the term with the smallest scaling dominates, while for large enough size of the Trotter step, the larger exponent dominates. The transition is controlled by the number of measurements, something that we show numerically in \cref{fig:ising_1x16_concentration} (left), where we indeed see that the crossover point moves to the left logarithmically (in a log-log plot) as the number of samples increases.

In \cref{fig:ising_1x16_concentration}, we analyse the error of \FAST{} as the number of layers is increased effectively moving the transition point to the right.
When simulating \FAST{} with multiple layers $N$, for each layer $t$ is scaled accordingly to $t/N$ such that the total evolution time is $t$. The circuit $S^{(l)}_k(t/N) e^{i\theta_{mr}(t/N)P_r}$ is then repeatedly applied $N$ times with a different $m$ and $r$ sampled at each step.

\begin{figure}[h!]
    \centering
    \includegraphics[width=\textwidth]{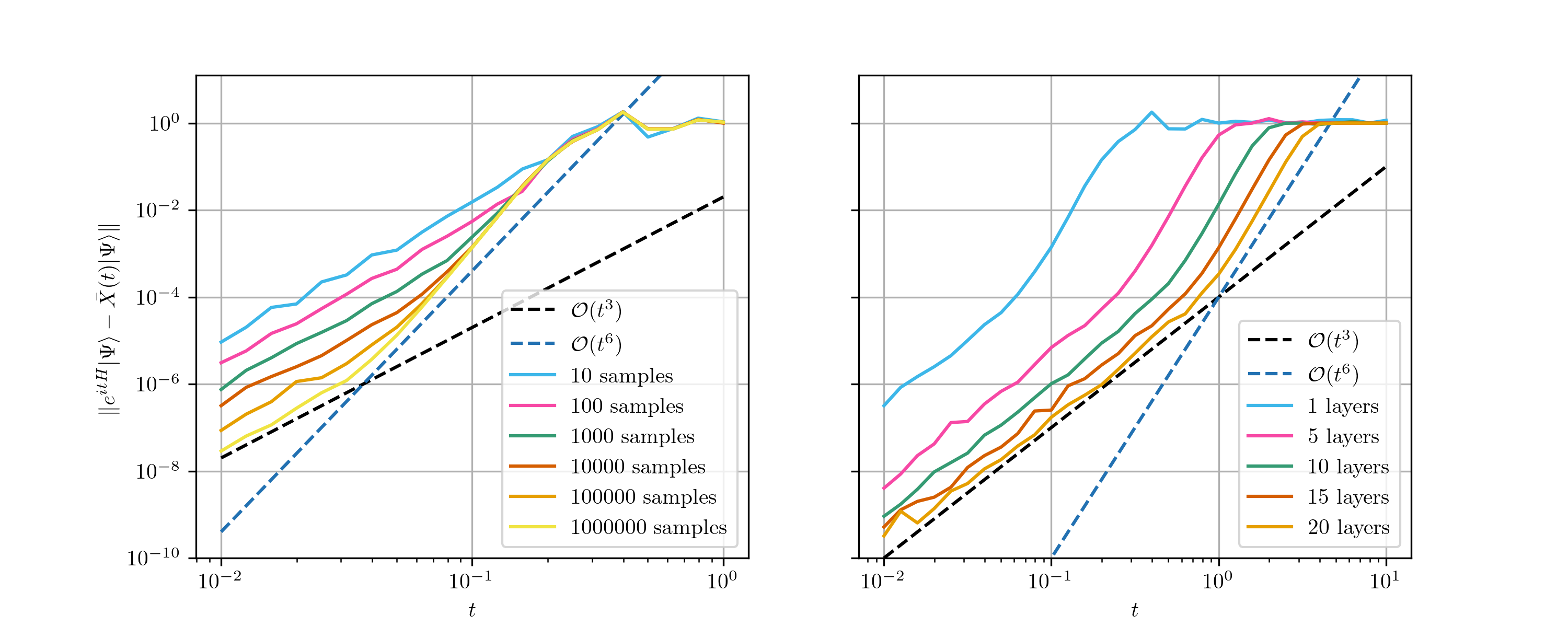} 
    \caption{Approximation error in $1 \times 16$ TF Ising model running the standard $k=2$ \FAST{} algorithm. This plot demonstrates how the concentration properties change as the number of samples (left), or number of layers (right) varies. Note the match with the predicted scaling showcased in \cref{coro:fluct} and \cref{fig:scaling_sketch}}
    \label{fig:ising_1x16_concentration}
\end{figure}

We also compare the performance of standard product formulas of order $k=2,4$ with the corresponding \FAST{} algorithm, to understand for which level of error \FAST{} outperforms standard product formulas, and by how much. In \cref{fig:alg_comparison_ising_1x16} we compare the performance of the different algorithms for fixed gate depths. The 2-qubit gate depth for one layer of Trotter 2 applied to the 1D Ising model is two -- we take $A=\sum_i Z_i$ and $B=\sum_{ij} X_iX_j$ in $e^{-itA/2}e^{-itB}e^{-itA/2}$ and use the gate depth of $e^{-itB}$. To estimate the 2-qubit gate depth of \FAST{}, we add the cost of implementing the longest Pauli string $P$ in $\Sigma_j$ to the Trotter cost. The 2-qubit depth of this additional rotation is $2\lceil \log_2 |P|\rceil - 1$ where $|P|$ is the Pauli string length. Therefore the resource estimates in \cref{fig:alg_comparison_ising_1x16} represent the worst case scenario for \FAST{}.

Here we are concerned with the amount of resources needed to implement the algorithms, so it is necessary to make several comments. For a budget of 2-qubit gate depth, it is not possible to implement some algorithms, in particular if the budget is small, something very relevant for NISQ applications. As the budget increases, other algorithms can be implemented. In \cref{fig:alg_comparison_ising_1x16} (left) we see that in high error regime and low budget, a second order formula works the best, for rather large time-steps $t$. As the time step is reduced, better accuracy is achieved. The white region in the lower right corner in \cref{fig:alg_comparison_ising_1x16}(left and right) corresponds to an error comparable with the maximum error between two arbitrary unitaries, so we dont include it.

\begin{figure}[h!]
    \centering
    \includegraphics[width=\linewidth]{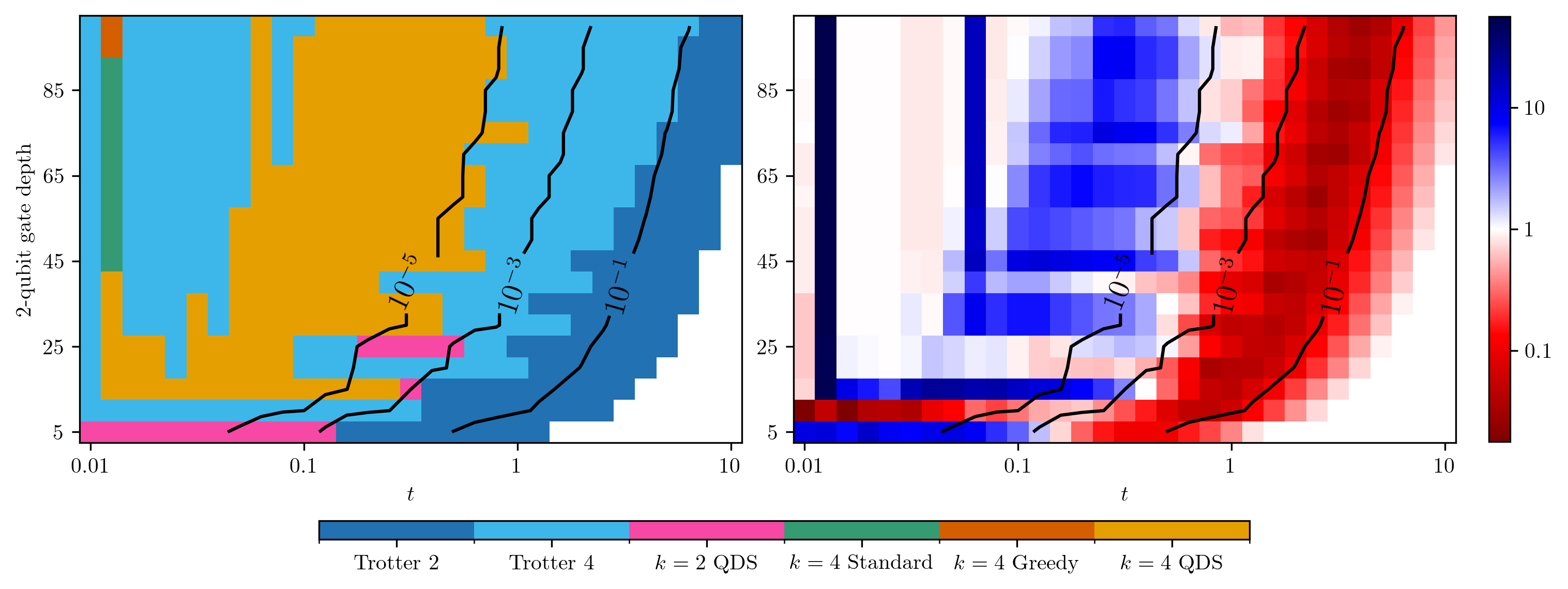}
    \caption{Performance comparison for Trotter 2, 4 and all variants of \FAST{} for $k=2$ and $k=4$ in the  $1 \times 16$ TF Ising model for different target times. For each row we fix a target gate depth, calculate how many layers of each algorithm can fit within that gate depth and compare the operator norms achieved to find the best algorithm. (Left) This figure shows which algorithm achieved the best operator norm. Note that although all variants were run, we only plot those which got the best results. The black contour lines show the values of the operator norm. (Right) Ratio between the best Trotter operator norm  vs the the best \FAST{} operator norm. Values larger than one indicate a better performance of \FAST{}. Any data which had an operator norm $> 0.9$ has been removed.}
    \label{fig:alg_comparison_ising_1x16}
\end{figure}

These error results and the scaling with different parameters follow the predicted bounds \cref{coro:fluct}. In order to assess the improvement that these formulas present over standard approaches, we also study the scaling of the number of gates needed to achieve a fixed error $\epsilon=10^{-3}$ in the one dimensional TF Ising model, as a function of the system's size $n$ and for an evolution time that scales $t=n$. This type of benchmark has been used \cite{Childs2018} to assess the level of resources needed to achieve a computation that is believed to be beyond what is classically simulable.

To produce the results in \cref{fig:ising_target_error_scaling} we made use of the fermionic linear optics library FLOYao.jl~\cite{floyao} to simulate the Trotter and \FAST{} circuits for sizes beyond classical simulability. A binary search algorithm was employed to find the number of layers required to reach the target error. We restrict the results to \FAST{} with $k=2$ and pick the worst (Standard) and best (Greedy + QDS) variants to bound the region of performance for \FAST{}. We also plot lines of best fit which to two decimal places are as follows, Trotter 2: $27.47n^{1.54}$, \FAST{} standard: $7.15n^{1.39}$, \FAST{} greedy + QDS: $6.89n^{1.22}$. 

\begin{figure}[h!]
    \centering
    \includegraphics[width=0.5\textwidth]{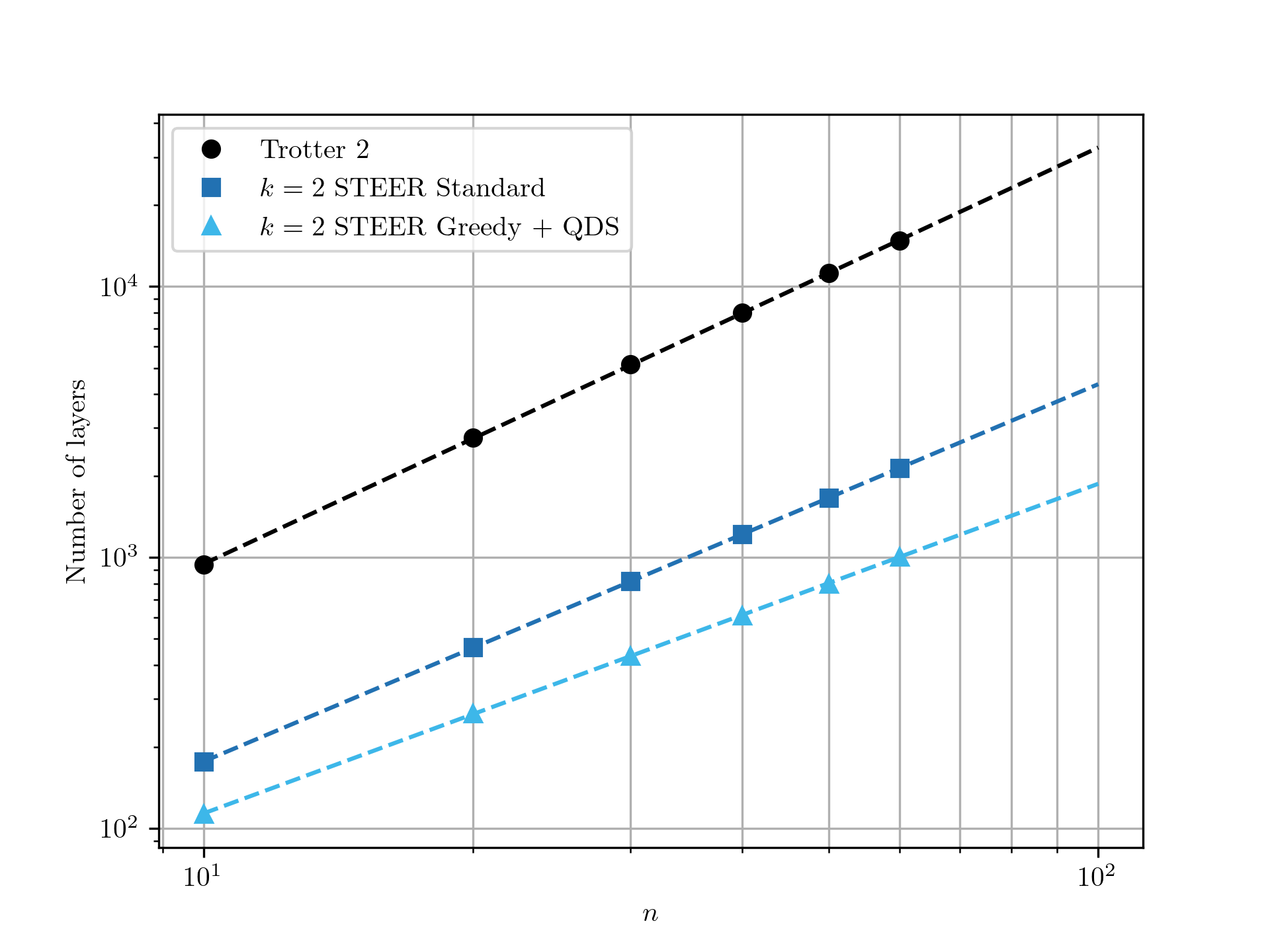} 
    \caption{Number of layers to reach a target error of $10^{-3}$ in the time evolution of a randomly sampled state of the $1 \times n$ TF Ising model. The evolution time is chosen to scale with the system's size $t = n$. The number of STEER samples taken was 10,000.}
    \label{fig:ising_target_error_scaling}
\end{figure}

\subsubsection{2D TF Ising model}

In 2D the TF Ising model is not unitarily equivalent to a free fermion model. We study the performance of \FAST{} in this model as it has been a natural playground to explore the usefulness of different algorithms \cite{Kim2023-yg,Bosse_2025}. We consider a $4\times 4$ square lattice with open boundary conditions. \cref{fig:ising_4x4_variants} shows the scaling of the error of our approximants  with respect to the exact simulation. As discussed in the previous section, it is clear the transition of scaling as a function of system's size, according to \cref{coro:fluct}.

\begin{figure}[h!]
    \centering
    \includegraphics[width=\textwidth]{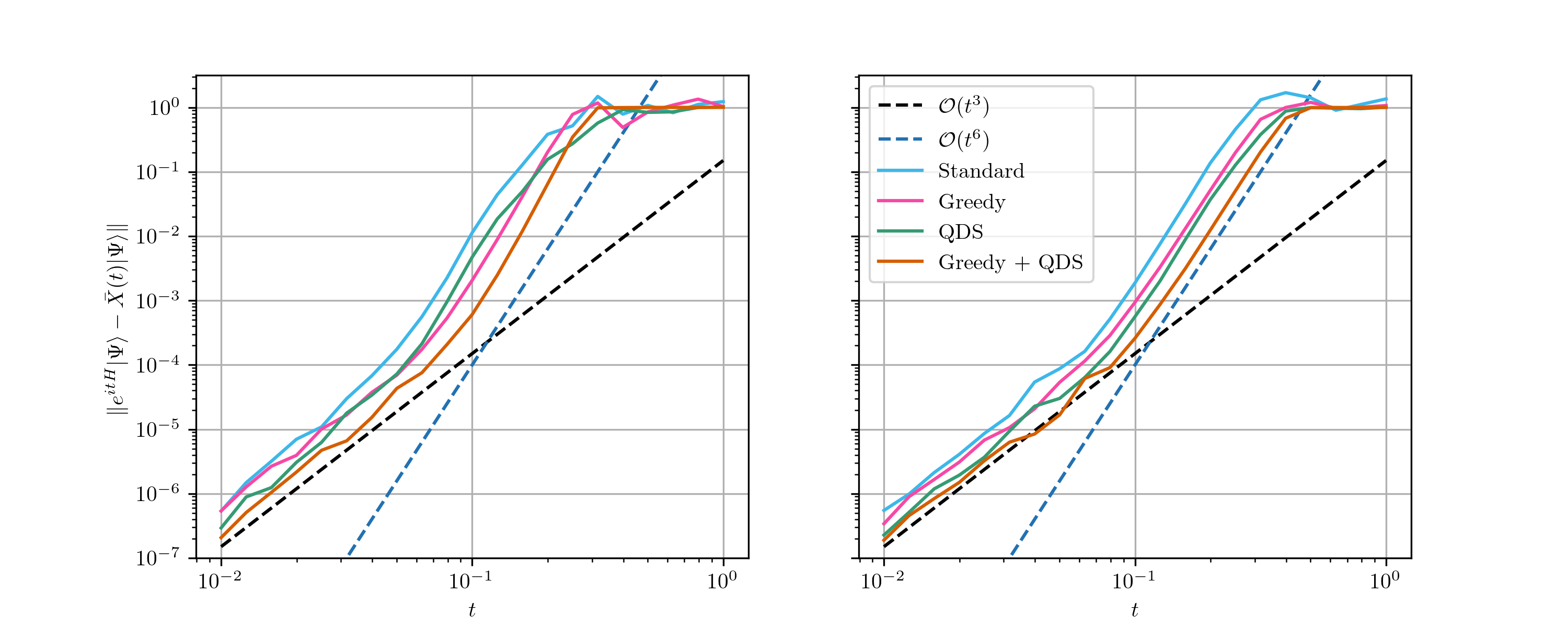} 
    \caption{Approximation error for $4 \times 4$ TF Ising model with all \FAST{} variants taking 10,000 samples. From left to right: \FAST{} with $k=2$ and symmetric \FAST{}.}
    \label{fig:ising_4x4_variants}
\end{figure}

\begin{figure}[h!]
    \centering
    \includegraphics[width=\linewidth]{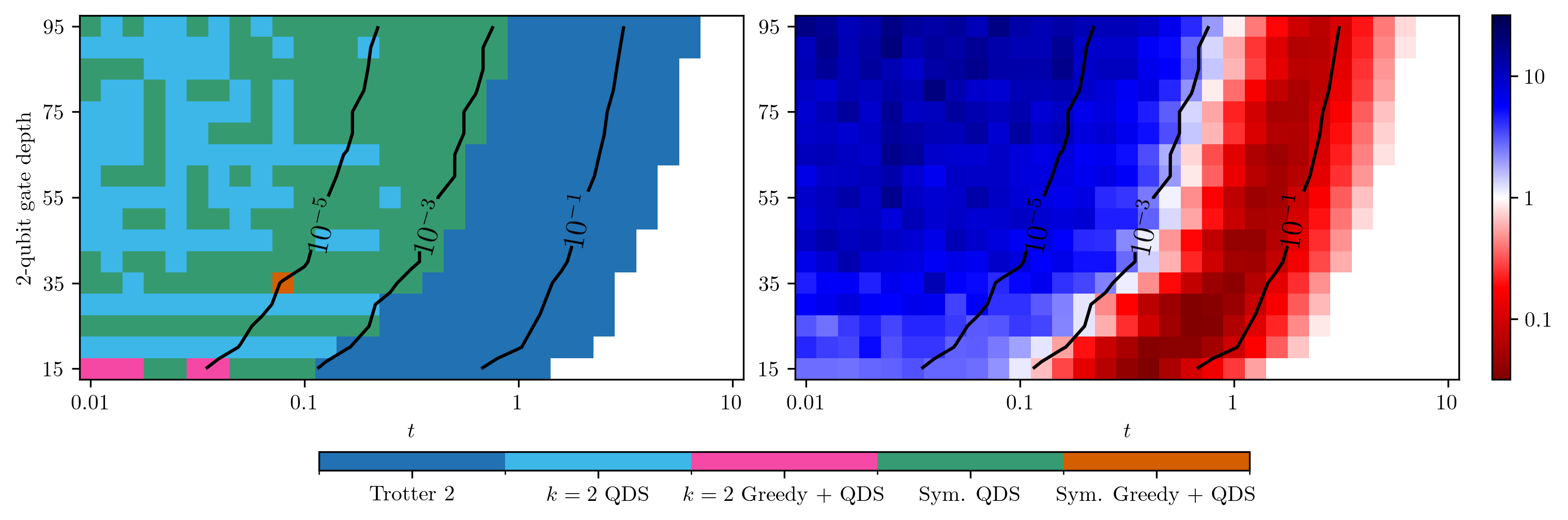}
    \caption{Performance comparison in the $4 \times 4$ TF Ising model with a range of $t$ values simulated with Trotter 2 and all variants of $k=2$ and symmetric \FAST{}. (Left) Algorithm with the best performance in terms of error in operator norm with respect to the exact evolution. (Right) Ratio of the best Trotter operator norm vs the best \FAST{} operator norm. Values larger than one indicate a best performance of \FAST{}.} 
    \label{fig:alg_comparison_ising_4x4}
\end{figure}

\subsection{1D Heisenberg model with random field}

The next model we simulate is the 1D Heisenberg model with a random field:
\begin{equation} \label{eq:heisenberg}
    H_\text{Heisenberg}=\sum_i (X_iX_{i+1} + Y_iY_{i+1} + Z_iZ_{i+1}) + \sum_i h_i Z_i
\end{equation}
where $X_i, Y_j, Z_j$ are the Pauli operators on site $i$. During our simulations we take $h_i \in [-1, 1]$ randomly sampled from a uniform distribution. For each simulation we also choose random computational basis states to evolve. 

\begin{figure}[h!]
    \centering
    \includegraphics[width=\textwidth]{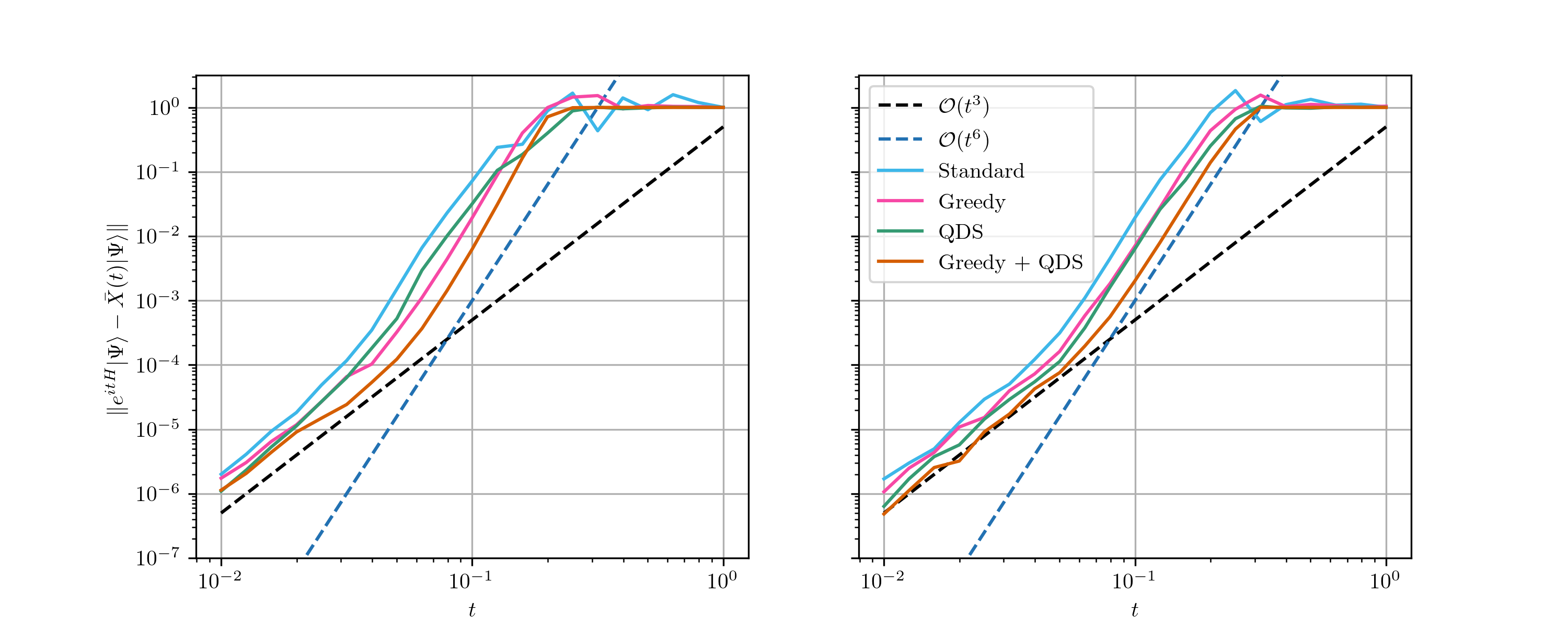} 
    \caption{Approximation error in the $1 \times 16$ Heisenberg model with all \FAST\,variants taking 10,000 samples. From left to right: \FAST\,with $k=2$ and symmetric \FAST.}
    \label{fig:heisenberg_1x16_variants}
\end{figure}

\begin{figure}[h!]
    \centering
    \includegraphics[width=\linewidth]{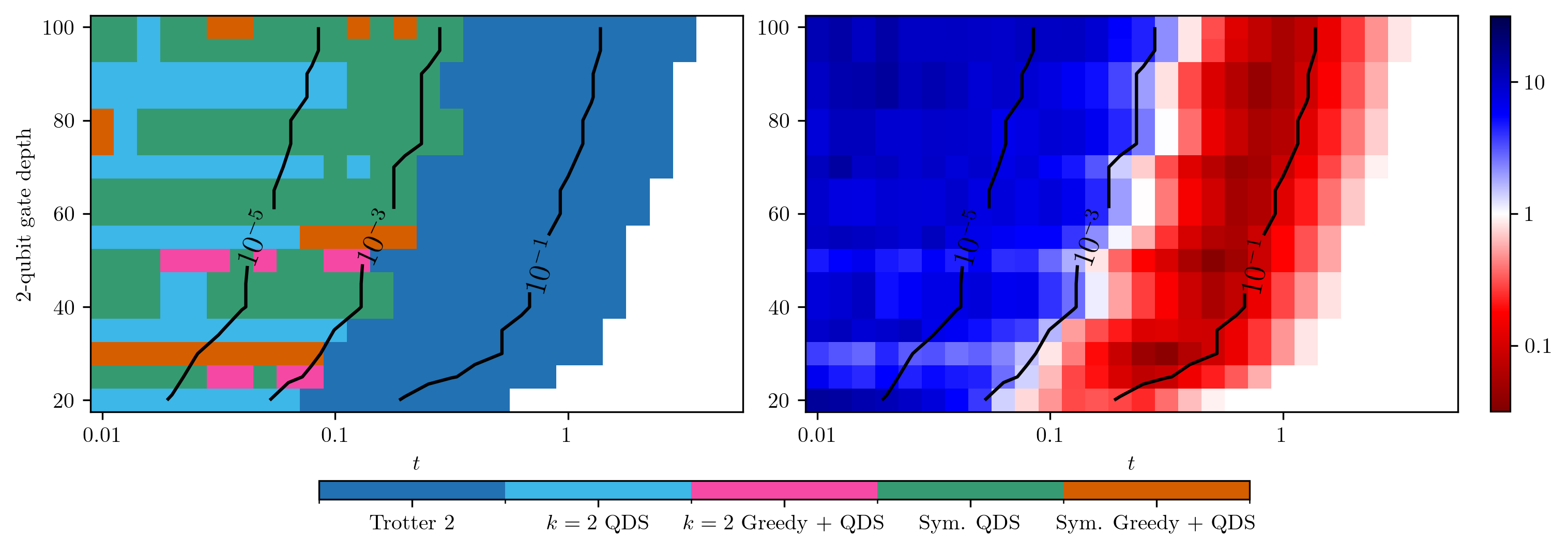}
    \caption{Performance comparison in the $1 \times 16$ Heisenberg model with a range of $t$ values simulated with Trotter 2 and all variants of $k=2$ and symmetric \FAST. (Left) Algorithm with the best performance in terms of error in operator norm with respect to the exact evolution. (Right) Ratio of the best Trotter operator norm vs the best \FAST{} operator norm. Values larger than one indicate a best performance of \FAST{}.}
    \label{fig:alg_comparison_heisenberg_1x16}
\end{figure}

\clearpage

\subsection{2D Fermi-Hubbard model}

The previous systems that we have looked at have been spin systems. Here we present results for the Fermi-Hubbard model which is a compact fermionic system that capture some of the physics of cuprates \cite{Arovas_2022}. The Hamiltonian is defined as:
\begin{equation}
    H_{\text{Hubbard}} = -J \sum_{\langle i, j \rangle, \sigma} \left( c^\dagger_{i\sigma} c_{j\sigma} + c^\dagger_{j\sigma} c_{i\sigma} \right) + U\sum_i n_{i\uparrow} n_{i\downarrow} 
\end{equation}
where $J$ is the tunneling amplitude, $U$ is the Coulomb potential, $c^\dagger_{i\sigma}/c_{i\sigma}$ are the fermionic creation/annihilation operators for spin $\sigma \in \{\uparrow, \downarrow\}$ on site $i$ and $n_{i\sigma} = c^\dagger_{i\sigma}c_{i\sigma}$ is the number operator. For the simulations we fix $U/J = 4$ and apply the Jordan-Wigner transformation to map to qubits which transforms the Hamiltonian terms as $c^\dagger_{i} c_{j} + c^\dagger_{j} c_{i} \mapsto \frac{1}{2} (X_i X_j + Y_i Y_j)Z_{i+1}\cdots Z_{j-1}$ and $n_in_j \mapsto \frac{1}{4} (I - Z_i)(I - Z_j)$. For the initial state that we evolve, we take the N\'eel state which consists of alternating up and down electrons in a checkerboard pattern, but with the electron in the center removed.
We study this state because the physics of (hole)-doping around half filling is expected to give rise to interesting phenomena like d-wave superconductivity, at low temperatures. The checkerboard (Neel) pattern was chosen as initial state due to its simplicity in the computational basis and relatively small energy (around 20\% from the bottom of the spectrum).
\begin{figure}[h!]
    \centering
    \includegraphics[width=0.9\textwidth]{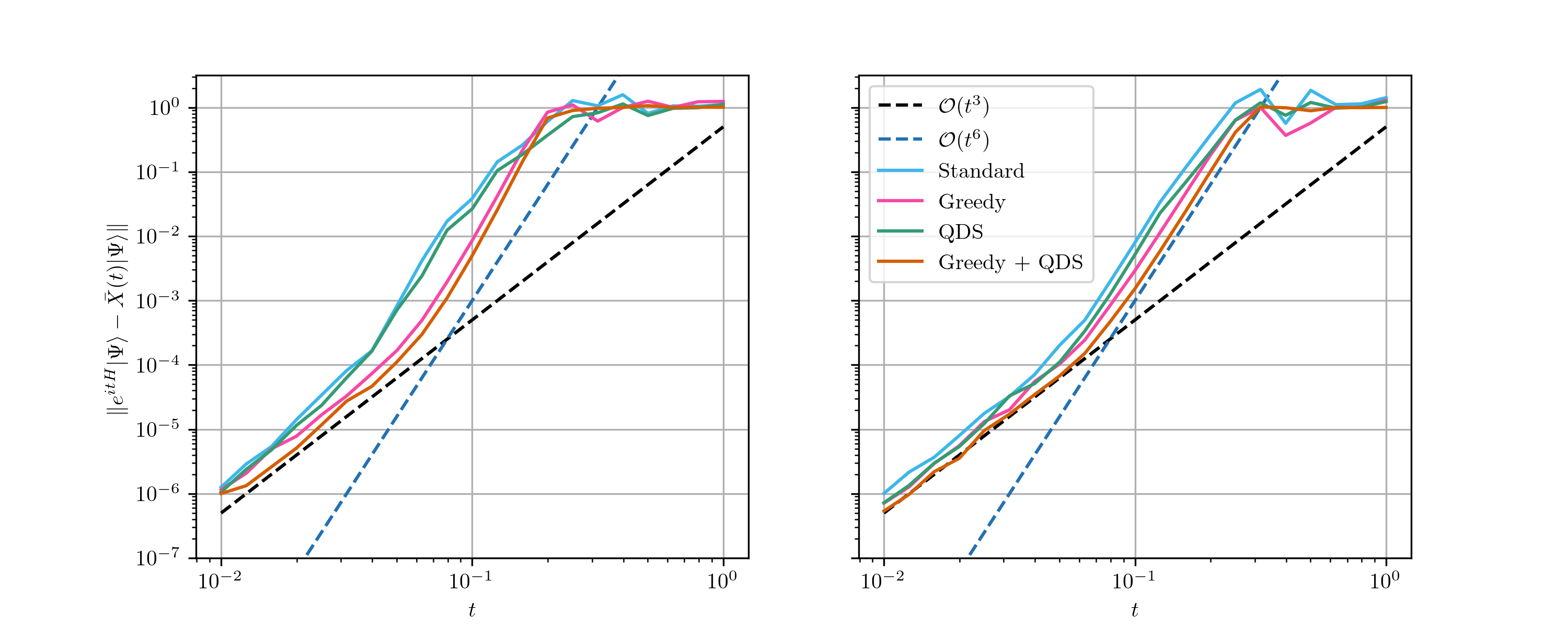} 
    \caption{Approximation error in $3 \times 3$ Hubbard model with all \FAST{} variants taking 10,000 samples. From left to right: \FAST{} with $k=2$ and symmetric \FAST.}
    \label{fig:hubbard_3x3_variants}
\end{figure}

\begin{figure}[h!]
    \centering
    \includegraphics[width=0.85\linewidth]{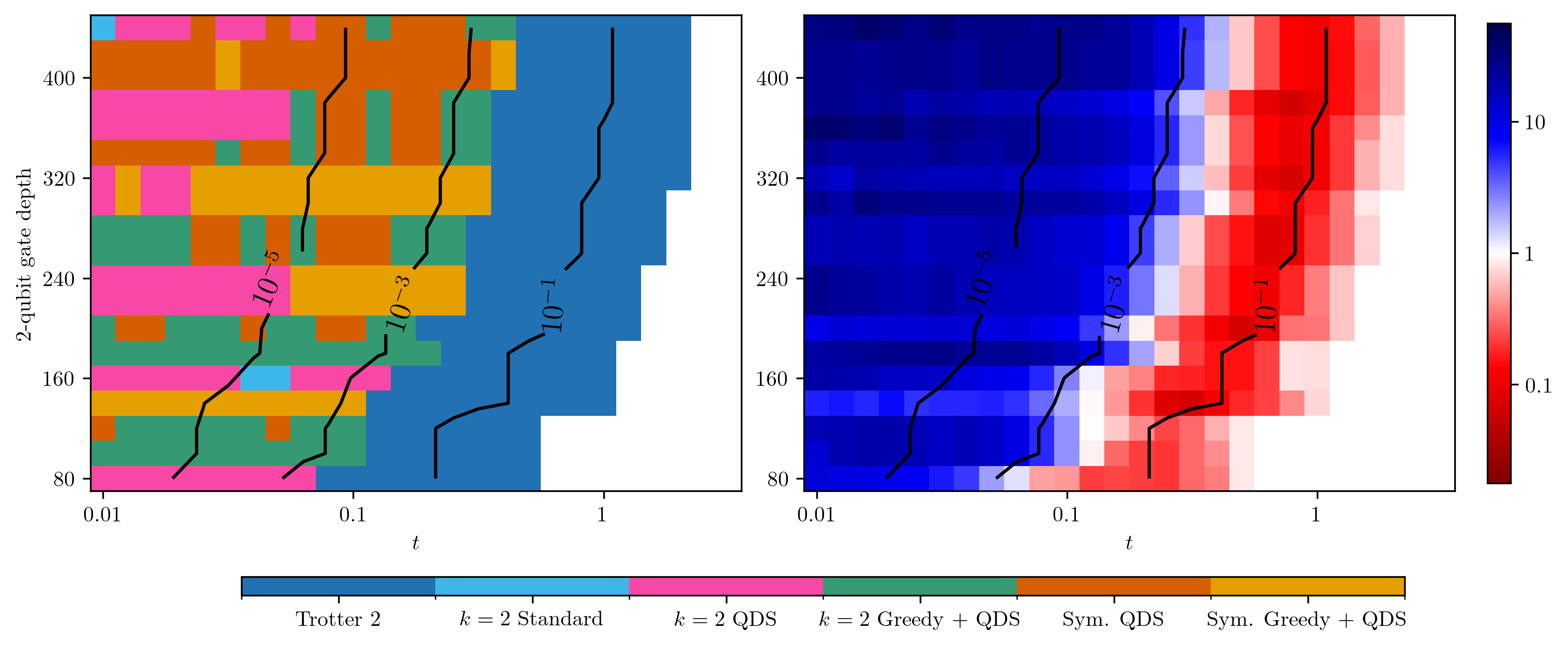}
    \caption{Performance comparison in the $3 \times 3$ Hubbard model with a range of $t$ values simulated with Trotter 2 and all variants of $k=2$ and symmetric \FAST. (Left) Algorithm with the best performance in terms of error in operator norm with respect to the exact evolution. (Right) Ratio of the best Trotter operator norm vs the best \FAST{} operator norm. Values larger than one indicate a best performance of \FAST{}.}
    \label{fig:alg_comparison_hubbard_3x3}
\end{figure}

\clearpage

\subsection{Hydrogen chain}\label{sec:hchains}

Another fermionic example that we investigate is the 1D Hydrogen chain where Hydrogen atoms are placed in a line $r$ Angstroms apart from each other. For the simulations in this section we use a chain of 4 Hydrogen atoms $r=0.4$ \AA{} apart, which lies just before the equilibrium point. The Hamiltonian is generated using PySCF~\cite{pyscf} in the STO-3G basis set. To run the algorithms we map the fermionic Hamiltonian to qubits using the Jordan-Wigner transformation -- once any terms with a magnitude smaller than $10^{-7}$ have been filtered out the Hamiltonian consists of 184 Pauli terms.

\begin{figure}[h!]
    \centering
    \includegraphics[width=\textwidth]{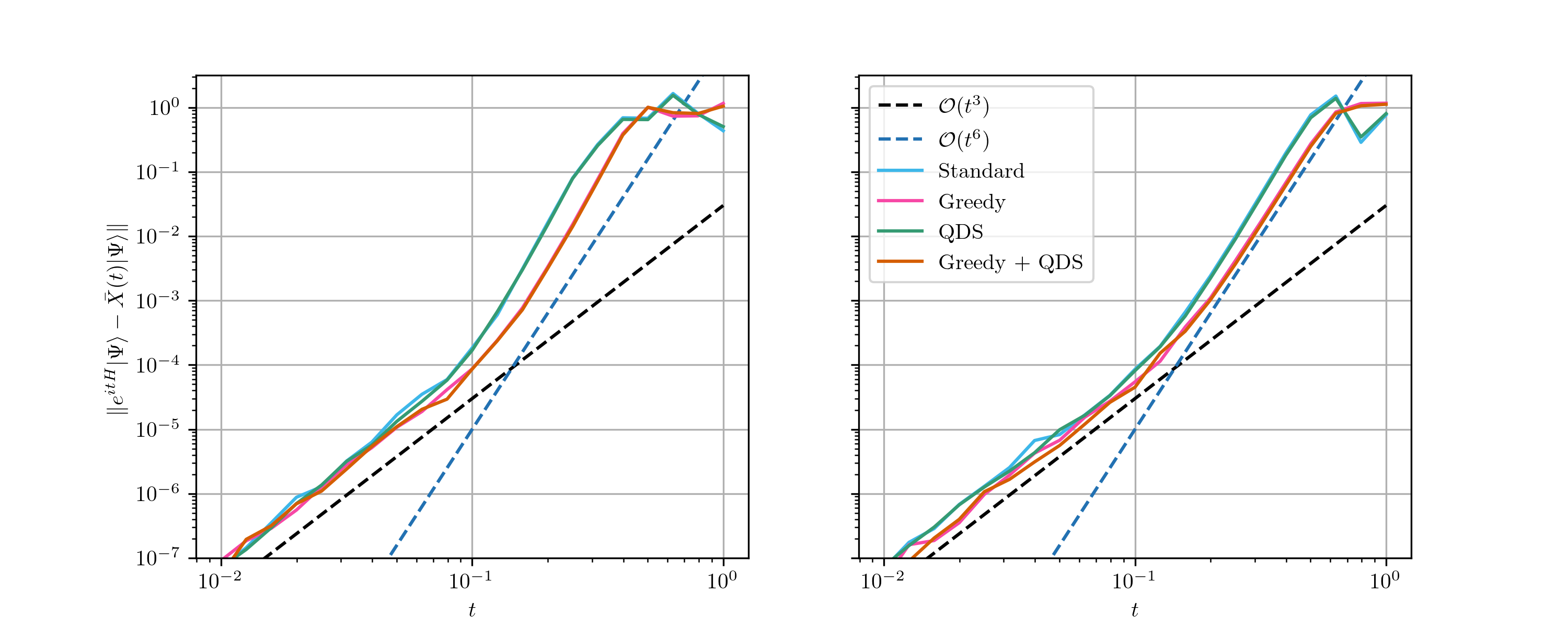} 
    \caption{Approximation error for the $1 \times 4$ Hydrogen chain with all \FAST{} variants taking 10,000 samples. From left to right: \FAST{} with $k=2$ and symmetric \FAST.}
    \label{fig:hubbard_3x3_variants}
\end{figure}

\begin{figure}[h!]
    \centering
    \includegraphics[width=\linewidth]{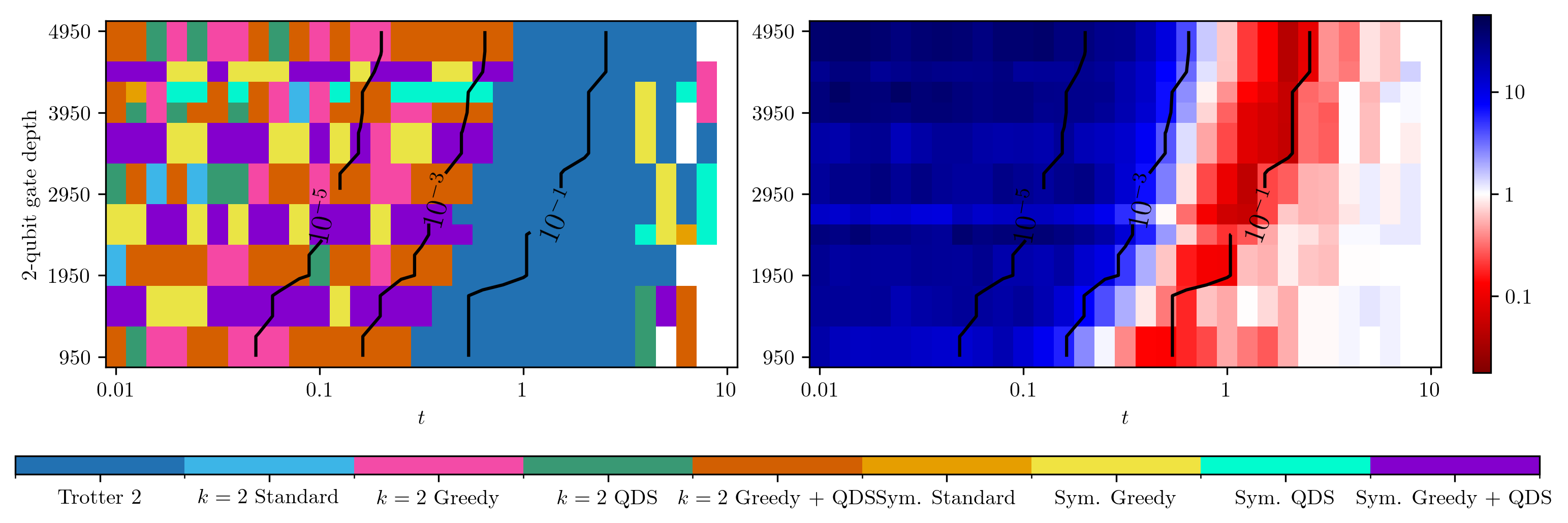}
    \caption{Performance comparison for  a $1 \times 4$ Hydrogen chain with a range of $t$ values simulated with Trotter 2 and all variants of $k=2$ and symmetric \FAST. (Left) Algorithm with the best performance in terms of error in operator norm with respect to the exact evolution. (Right) Ratio of the best Trotter operator norm vs the best \FAST{} operator norm. Values larger than one indicate a best performance of \FAST{}.}
    \label{fig:alg_comparison_h4}
\end{figure}

\subsection{Circuit depth cost}\label{sec:depth_diss}

The circuit reduction achieved by \FAST{} can be estimated by considering the circuit necessary to implement a $2k$-order product formula, based on a repeated application of a $k$-order formula. Let's define the circuit depth needed to implement a $k$-order formula as $s_k$. Then, using the recursive definition of the higher order product formulas \cite{Susuki1991}, we can bound $s_{k+2} = 5s_k$. This immediately implies
\begin{align}
    s_{2k}=5^{k/2}s_k.
\end{align}
The factor of $5^{k/2}$ can be improved in some cases \cite{morales2022greatly}, but the growth is always $\exp({O(k)})$. In contrast, \FAST{} produces a circuit with depth upperbounded by $s^{\FAST{}}_{2k}=s_k +2\log_2(n)$, where $n$ is the number of qubits. The logarithmic contribution appears by considering the worst case scenario of sampling a Pauli from the effective error Hamiltonian \cref{eq:A_k}, with a support over the whole system. Using a ladder construction \cite{Clinton_2024} an arbitrary Pauli term acting non-trivially over $k$ qubits can be implemented with depth $2\log_2(k)=1$. This in turn can be bounded by the size of the system. This scaling indicates that for increasing order, the gain in circuit complexity is exponential in the number of layers.

\section{Discussion}\label{sec:discussion}
We have presented a method for improving the error scaling of Hamiltonian simulation by combining conventional approximants with a sampling protocol for the error unitary, which we dub \FAST{}. This algorithm demonstrates a quadratic factor improvement over the approximant from which is constructed, with only logarithmic circuit-depth overhead and without the need for ancilla qubits. While there exist product formulas with the same scaling as \FAST{}, the gate complexity of our approach is {\it exponentially} better (in the approximation order) than the corresponding product formula (see discussion in \cref{sec:depth_diss}). We achieve this by characterising the generator of the error unitary and sampling from it. We give an explicit construction of the terms to sample, and analyse the error of the algorithm with respect to the exact channel. Using fluctuation bounds, we are able to bound the effect of performing multiple measurement rounds in this algorithm. This highlight a difference between two usual computational paradigms in quantum simulation, the ``single-shot" vs the ``many-shots" model. In the first case, a bound on the diamond norm between the approximate and the exact channels is enough to understand the performance of the algorithm. In the second case, the effect of the repeated rounds of measurements can be understood by bounding the fluctuations of the channel with respect to the average channel. We study this explicitly, and show numerical results that demonstrate the theoretical scaling that we obtain. The main result of this analysis appears in \cref{coro:fluct}, it is sketched in \cref{fig:scaling_sketch} and shown in practice in \cref{sec:numerics}.
In doing this we have further shown that the sample mean of the random unitary concentrates exponentially about its average and so derive rigorous sampling complexity guarantees for our procedure.

Extensive numerical simulations on a range of different systems corroborates our findings. The analysis of errors on practical terms, compared with standard product formulas is discussed in \cref{sec:numerics}. We also provide resource estimates for the simulation of the 1D Heisenberg model with random fields, that has been proposed as a natural benchmark for quantum computers beyond the classically tractable regime \cite{Childs2018}. Our results suggest that simulating this model up to a time $T=100$ in a 100 qubit system, with error $\epsilon = 10^{-3}$ would require $\sim10^4$ layers with a CNOT gate depth of 14 per layer.

As an added bonus, our results illustrate a mathematical connection between product-formula error unitaries and Zassenhaus decompositions (see \cref{sec:general_zas}), which may be of independent interest.

\subsection*{Acknowledgments}

We are grateful to Andrew Childs and Raul Garcia-Patron for very helpful discussions and their feedback on this work. We also thank Ashley Montanaro for his very useful suggestions on this manuscript.

\clearpage

\appendix

\begin{center}  
\LARGE \textbf{APPENDIX}
\end{center}

\section{Explicit General partitions}\label{sec:general_partition}

Here we provide explicit expressions for more general partitions where the Hamiltonian $H=\sum_k h_k$ is composed of arbitrary number of summands. We start with the second order product formula
\begin{align}
    {S}_{2}^{(l)}(t)&=\prod_{m=1}^{l}e^{i\frac{t}{2}h_{m}}\prod_{m=l}^{1}e^{i\frac{t}{2}h_{m}}
\end{align}
where we make explicit the number of summands in the Hamiltonian.
From this formula, applying the extension of the definition of the effective Hamiltonian
\begin{align} \label{eq:A2}
    A_2^{(l)}(t)&:=-i\frac{d}{dt}(S_{2}^{(l)\dagger}(t)U)U^{\dagger}S_{2}^{(l)}(t)=S^{(l)\dagger}_2(t) H S^{(l)}_2(t) -i\frac{d}{dt}(S_{2}^{(l)\dagger}(t))S_{2}^{(l)}(t)
\end{align}
we find
\begin{align}\label{eq:A2gen}
    A_2^{(l)}:=S^{(l)\dagger}_2(t) H S^{(l)}_2(t) -\frac{1}{2}\sum_{k=1}^{l}\prod_{m=1}^{k}e^{-i\frac{t}{2}h_{m}}\left(h_{k}+\prod_{m=k+1}^{l}e^{-i\frac{t}{2}h_{m}}\prod_{m=l}^{k}e^{-i\frac{t}{2}h_{m}}h_{k}\prod_{m=k}^{l}e^{i\frac{t}{2}h_{m}}\prod_{m=l}^{k+1}e^{i\frac{t}{2}h_{m}}\right)\prod_{m=k}^{1}e^{i\frac{t}{2}h_{m}}.
\end{align}

For an arbitrary product formula
\begin{equation}
    {S}^{(l)}(t)=\prod_{m=1}^{l}e^{it\theta_m h_{m}},
\end{equation}
equation~(\ref{eq:A2}) can be expanded as 
\begin{equation}\label{eq:A_gen}
    A^{(l)} := S^{(l)\dagger} H S^{(l)} - \sum_{k=1}^l \theta_k \left( \prod_{m=l}^{k+1} e^{-it \theta_m h_m} \right) h_k \left( \prod_{m=k+1}^{l} e^{it \theta_m h_m} \right).
\end{equation}

\cref{eq:A2gen} and \cref{eq:A_gen} are directly amenable for symbolic computations.

\section{General expansion of $A_2(t)$ for two summands in $H$}\label{app:genA2}

For future reference, we write explicitly the result of $A_2(t)$ for a partition $H=A+B$ consisting on just two terms,

\begin{align}
A_2(t)=\sum_{n=0,k=2}^{\infty}\frac{(-it)^{n+k}}{2^{n}n!k!}\left(\left(\frac{1}{2}-\frac{k}{2}\right)Ad_{A}^{n}(Ad_{B}^{k}(A))+\sum_{p=0}^{\infty}\frac{(-it)^{p}}{2^{k}p!}Ad_{A}^{n}(Ad_{B}^{p}(Ad_{A}^{k}(B)))\right)
\end{align}
where $Ad_X(Y):=[X,Y]$. The iterated operator is defined as $Ad^n_X(*)=Ad^{n-1}_X([X,*])$, with $Ad^0_X(Y):=Y$.
Up to order $t^4$, the Hamiltonian $A_2(t)$ is given by $A_2(t)=t^2\Omega_2+t^3\Omega_3+t^4\Omega_4+O(t^5)$, where explicitly
\begin{align}\label{eq:Sigmas}
    \Omega_2&=-\frac{1}{8}\left(Ad_{A}^{2}(B)-2Ad_{B}^{2}(A)\right)\\
    \Omega_3&=\frac{i}{12}\left(Ad_{A}^{3}(B)+\frac{3}{2}Ad_{B}\left(Ad_{A}^{2}(B)\right)-2Ad_{B}^{3}(A)-\frac{3}{2}Ad_{A}(Ad_{B}^{2}(A))\right)\\\nonumber
   \Omega_4&=\frac{1}{2^{4}}\left(\frac{11}{24}Ad_{A}^{4}(B)-Ad_{B}^{4}(A)-\frac{4}{3}Ad_{A}(Ad_{B}^{3}(A))-\frac{1}{2}Ad_{A}^{2}(Ad_{B}^{2}(A))+Ad_{B}^{2}(Ad_{A}^{2}(B))\right)\\
   &+\frac{1}{2^{4}}\left(\frac{1}{3}Ad_{B}(Ad_{A}^{3}(B)))+Ad_{A}(Ad_{B}(Ad_{A}^{2}(B)))\right)
\end{align}

\section{Variants of the algorithm}\label{app:other_vars}

Here we present two variants of the \FAST{} algorithm that differ in the way the sampling is done. 

\subsection{Pseudocode for $\mathcal{Z}_6$, arbitrary number of summands in $H$}\label{algo:var1}
\begin{algorithm}[H]
\SetAlgoLined
\KwResult{Stochastic approximant of $e^{itH}$ with error of order $O(t^6)$}
 Given a Hamiltonian $H=\sum_{k=1}^l h_k$, where $h_k$ are operators such that $e^{ith_k}$ is directly implementable, a second order product formula $S_2^{(l)}(t)$, a state $|\Psi\rangle$, and a evolution time $t$

\medskip

 Compute the series expansion in time of \cref{eq:A2gen} up to order $t^5$, i.e
\begin{align}
A_2^{(l)}(t)=t^2\Omega_2+t^3\Omega_3 + t^4 \Omega_4 +O(t^5)
\end{align}

\For{$j=2..4$}{
Compute $\Omega_j$ as a sum of Pauli terms $\Omega_j=\sum_k\alpha_{jk}P_k$\\
Output the probability distribution $p_{jk}:=\frac{|\alpha_{jk}|}{\sum_{p}|\alpha_{jp}|}$
} 
\textit{Variant one: One term sampling\\}
Define the probability distribution $p_m(t):=\frac{t^{m}}{m+3}\left(\frac{1}{3}+\frac{t}{4}+\frac{t^{2}}{5}\right)^{-1}$

 \For{$j =1\dots N_{samples}$}{
 Sample the random variable $M\in\{0,1,2\}$ from the probability distribution $p_m(t)$
 , i.e ${\rm Pr}(M=m)=\frac{t^{m}}{m+3}\left(\frac{1}{3}+\frac{t}{4}+\frac{t^{2}}{5}\right)^{-1}$\\
 Sample the random variable $K$ from the probability distribution $p_{m+3,k}$, i.e ${\rm Pr}(K=k)=\frac{|\alpha_{m+3,k}|}{\sum_{p}|\alpha_{m+3,p}|}$.\\
 Construct the estimator
 $X^{(j)}(t)=S_2^{(l)}(t)e^{i\theta_{mk}(t)P_k}|\Psi\rangle$
 where $\theta_{mk}(t)={\rm sign}(\alpha_{mk})|(\frac{t^{3}}{3}+\frac{t^{4}}{4}+\frac{t^{5}}{5})\sum_{p}|\alpha_{mp}|$
 }

\textit{Variant two: Each order in the time series is sampled}

 \For{$j =1\dots N_{samples}$}{

 Sample the random variables $K_m,$ $(m=2\dots 4)$ from the respective probability distribution $p_{m,k}$, i.e ${\rm Pr}(K_m=k)=\frac{|\alpha_{m,k}|}{\sum_{p}|\alpha_{m,p}|}$.\\
 Construct the estimator
 $X^{(j)}(t)=S_2^{(l)}(t)\prod_{m=2}^4 e^{i\theta_{mk}(t)P_k}|\Psi\rangle$
 where $\theta_{mk}(t)=\frac{t^{m+1}}{m+1}{\rm sign}(\alpha_{m,k})\sum_{p}|\alpha_{mp}|$
 }
 
 Output the sample mean $\bar{X}(t)=\frac{1}{N_{samples}}\sum_j X^{(j)}(t)$\\

 \caption{Approximating time evolution with stochastic formulas}
\end{algorithm}

\printbibliography

\end{document}

%% file: preamble2.tex
\usepackage[a4paper, total={7in, 9in}]{geometry}
\usepackage[utf8]{inputenc}
\usepackage{tikz}
\usepackage{background}
\usepackage{caption}
\usepackage{graphicx}
\usepackage{subfig}
\usepackage[normalem]{ulem}

\usepackage{longtable}

\backgroundsetup{
   scale=1,
   angle=0,
   opacity=1,
   color=black,
   contents={\begin{tikzpicture}[remember picture, overlay, font=\sffamily]
     \end{tikzpicture}}
 }

\usepackage[
    backend=bibtex,
    style=numeric,
    sorting=none
]{biblatex}
\bibliography{references}

\usepackage{amsmath}
\usepackage{braket}
\usepackage{amsthm}

\usepackage{listings}
\usepackage[ruled,vlined]{algorithm2e}
\usepackage{amssymb}
\usepackage[ampersand]{easylist}
\usepackage{enumitem}
\usepackage{hyperref}
\hypersetup{
  colorlinks = true, 
  linkcolor={blue!50!black}, 
  citecolor=red, 
  urlcolor=blue, 
  linktoc=page,
}
\usepackage[capitalise, nameinlink]{cleveref}
\usepackage[
  style=long,
  nolist, 
  nonumberlist,
  nopostdot,
  acronym,
  toc=true, 
]{glossaries}
\usepackage{bm}
\usepackage{booktabs}
\usepackage{multirow}

\usepackage{array}
\usepackage[british]{babel}

\newtheorem{theorem}{Theorem}
\newtheorem{definition}{Definition}

\newtheorem{corollary}[theorem]{Corollary}
\newtheorem{fact}{Fact}

\newtheorem{innercustomresult}{Result}

%% file: macros.tex
\newcommand{\be}{\begin{equation}}
\newcommand{\ee}{\end{equation}}

